\documentclass[a4paper,11pt]{article}
\usepackage{amsfonts,amsthm,amssymb}
\usepackage{dsfont}
\usepackage{tikz}
\usepackage[affil-it]{authblk}
\newcommand\COMP{\hbox{C\kern -.58em {\raise .54ex \hbox{$\scriptscriptstyle |$}}
\kern-.55em {\raise .53ex \hbox{$\scriptscriptstyle |$}} }}
\newcommand\NN{\hbox{I\kern-.2em\hbox{N}}}
\newcommand\RR{\hbox{I\kern-.2em\hbox{R}}}
\newcommand\sRR{{\it \hbox{I\kern-.2em\hbox{R}}}}
\newcommand\QQ{\hbox{I\kern-.53em\hbox{Q}}}
\newcommand\PP{\hbox{I\kern-.53em\hbox{P}}}
\newcommand\EE{\hbox{I\kern-.53em\hbox{E}}}
\newcommand\ZZ{{{\rm Z}\kern-.28em{\rm Z}}}
\newcommand\be{\begin{equation}}
\newcommand\ee{\end{equation}}
\newtheorem{theorem}{Theorem}[section]

\newtheorem{proposition}[theorem]{Proposition}
\newtheorem{remark}[theorem]{Remark}

\newtheorem{example}[theorem]{Example}
\newtheorem{lemma}[theorem]{Lemma}

\newtheorem{definition}[theorem]{Definitions}
\newtheorem{corollary}[theorem]{Corollary}

\newcommand\beq{\begin{eqnarray}}
\newcommand\eeq{\end{eqnarray}}
\newcommand\bq{\begin{eqnarray*}}
\newcommand\eq{\end{eqnarray*}}

\newcommand{\mm}{m}

\def\comg#1{\left ( #1\right )\!^{p,\mathbb G}}
\def\prof#1{ \phantom{l}^{p,\mathbb F}\!\left ( #1\right )}
\def\prog#1{ \phantom{l}^{p,\mathbb G}\!\left ( #1\right )}

\def \Lbrack {[\![}
\def \Rbrack {]\!]}

\def\mbf{\mathbb F}
\def\mbg{\mathbb G}

\makeatletter
\newcommand*\bigcdot{\mathpalette\bigcdot@{.5}}
\newcommand*\bigcdot@[2]{\mathbin{\vcenter{\hbox{\scalebox{#2}{$\m@th#1\bullet$}}}}}
\makeatother
\newcommand{\is}{\bigcdot }

\begin{document}
\title{Structure conditions under progressively added information}

\author{Tahir Choulli \thanks{corresponding author, {Email: tchoulli@ualberta.ca} }}

\affil{Mathematical and Statistical Sciences Dept.\\University of Alberta, Edmonton, Canada}
\author{Jun Deng}

\affil{School of Banking and Finance, \\University of International Business and Economics, Beijing, China}

\maketitle

\begin{abstract} It has been understood that the ``local" existence of the Markowitz' optimal portfolio or the solution to the local-risk minimization problem is guaranteed by some specific mathematical structures on the underlying assets price processes  known in the literature as ``{\it Structure Conditions}". In this paper, we consider a semi-martingale market model, and an arbitrary random time that is not adapted to the information flow of the market model. This random time may model the default time of a firm, the death time of an insured, or any the occurrence time of an event that might impact the market model somehow. By adding additional uncertainty to the market model, via this random time, the {\it structures conditions} may fail and hence the Markowitz's optimal portfolio and other quadratic-optimal portfolios might fail to exist. Our aim is to investigate the impact of this random time on the structures conditions from different perspectives. Our analysis
allows us to conclude that under some mild assumptions on the market model and the random time, these structures conditions will remain valid on the one hand. Furthermore, we provide two examples illustrating the importance of these assumptions. On the other hand, we describe the random time models for which these structure conditions are preserved for any market model. These results are elaborated separately for the two contexts of stopping with the random time and incorporating totally a specific class of random times respectively.
\end{abstract}
\section{Introduction}
Since the seminal work of Markowitz on the optimal portfolio, the quadratic criterion for hedging contingent claims becomes very popular and an important topic in mathematical finance, modern finance, and insurance. In this context, two main competing quadratic approaches were suggested. Precisely, the local risk minimization and the mean variance hedging. For more details about these two methods and their relationship, we refer the reader to Health et al. \cite{healthplaten2001},  Cerny and Kallsen \cite{cernykallsen2007}, Duffie and Richardson \cite{duffine1991},
Delbaen and Schachermayer \cite{delbaenScha1996},  Biagini et al. \cite{biaginietal2000}, Jeanblanc et al. \cite{jeanblancschweizer2012},   Schweizer \cite{Schweizer1995, Schweizer1999},   Choulli et al. \cite{choullivandaeles2010}, Laurent and Pham \cite{laurentPham1999} and the references therein. \\

 One important common feature for these method lies in the assumptions that the market model should fulfill in order that the two methods admit solutions at least locally. These conditions are known by the ``Structure Conditions" (called SC hereafter) and sound to be the alternative to non-arbitrage condition in this quadratic context. Indeed, for the case of continuous price processes, it is proved that these conditions are equivalent to the non-
arbitrage of the first kind (No-Unbounded-Profit-with-Bounded-Risk, called NUPBR hereafter), or equivalently to the existence of a local martingale deflator for the market model. For details about these equivalence, we refer the reader to Choulli and Stricker \cite{stricker1990}.  However, in the general case, the two concepts (i.e. SC and NUPBR) differ tremendously.\\

\noindent Recently, there has been an upsurge interest  in  investigating  the effect of different information levels on arbitrage theory and utility maximization problem, see  \cite{aksamit/choulli/deng/jeanblanc}, \cite{Choulli2013}, \cite{fjs},  \cite{kohatsusulem06},   \cite{ais98} and the references therein.   From the economic standpoint of view, information is a commodity that bears values; and economic agents desire information because it helps them to make decision and maximize their state-dependent utilities, especially when they are facing uncertainties. For more details about this economic views we refer the reader to Allen \cite{allen} and Arrow \cite{arrow73, arrow84, arrow99} and the references therein.\\

\noindent In this paper, we study the impact of some extra information/uncertainty on the Structure Conditions. This extra information comes from a random time $\tau$ that is not adapted to the public information represented by a filtration $\mathbb{F}:=({\cal F}_t)_{t\geq 0}$. There are two mainstreams to combine the information coming from $\tau$ and $\mathbb{F}$: The initial enlargement and progressive enlargement of the filtration $\mathbb{F}$ (see \cite{Jeu}, \cite{jacod1},  \cite{Yor} and the references therein). Herein, we restrict our attention to adding the information from $\tau$ progressively to $\mathbb{F}$, and the resulting larger filtration will be denoted throughout the paper by $\mathbb{G}$. In this paper, we
are
dedicated to investigate the following two questions:\\

 \centerline{ \hspace*{0.5cm} For which pair $(\tau, S)$, does $(S, \mathbb{G})$ satisfy SC?\hfill   {{\bf
(Prob1) } } }
\vskip 0.15cm
and
\vskip 0.35cm

\centerline{ \hspace*{0.5cm}For which model of $\tau$, $(X, \mathbb{G})$ fulfills SC as long as $(X, \mathbb F)$ does? \hfill   {{\bf
(Prob2) } } }
\vskip 0.35cm

\noindent To answer the two problems  {\bf(Prob1)} and  {\bf(Prob2)}, we split the time horizon $\Lbrack 0,+\infty \Lbrack$ into two disjoint intervals $\Lbrack 0,\tau \Rbrack$ and $\Rbrack \tau, +\infty\Lbrack$. In other words, we investigate the impact of $\tau$ on the structures conditions of $S$ by studying $(S^\tau,\mathbb{G})$  and $(S - S^\tau, \mathbb{G})$  individually.\\

\noindent This paper contains four sections, including the current section,  and an appendix where we recall some useful results. Section \ref{subsectPreliminaries} defines the mathematical model, while Section \ref{StructureConditions4Stau} addresses the structure conditions for models stopped at $\tau$. The last section, Section \ref{StructureConditionsAftertau}, deals with the ``part-after-$\tau$".
\section{The mathematical model and preliminaries}\label{subsectPreliminaries}
Our mathematical model starts with a stochastic basis $(\Omega, {\cal A},  {\mathbb  F }=({\cal F}_t)_{t\geq 0},  \mathbb P)$,  where ${\mathbb  F }$ is a filtration satisfying the usual conditions of right continuity and completeness and represents the flow of ``public" information over time. On this filtered probability space, we consider given $S$, a {\bf one-dimensional special semimartingale} such that 
\begin{equation}\label{S-DoobMeyerDecomposition}
 S=S_0+M+A,
\end{equation}
 where $M$ is a locally square integrable local martingale and $A$ is predictable with finite variation. $S$ models the tradeable risky asset.The extension to multi-dimension case and/or to general semimartingales is very doable at the expenses of some technicalities that we avoid herein for the sake of well illustrating the main ideas. Thus,  $(\Omega, {\cal A}, {\mathbb  F },  S, \mathbb P)$ constitutes the initial market model. In addition to this model, we consider an $\cal A$-measurable random time $\tau: \Omega \rightarrow \mathbb{R}_+$ that is fixed from the beginning and for the entire paper. This random time can represent the agent's death time, the sudden retirement time, the bankruptcy time of a firm, the default time, or any time of occurrence of an event that might affect the initial market and/or the agents. Mathematically, this random time is not a stopping time with respect to $\mathbb F$, and hence one can not confirm at time $t$ whether this random time happened or not yet. However, using the flow of information $\mathbb F$, we can observe the survival probability of this random time. To formulate this rigorously, we
associate to $\tau$ the pair $(D, \mathbb G)$ given by
 \begin{equation}\label{AandfiltrationG}
D:= I_{\Lbrack\tau,+\infty\Lbrack},\ \ \ \mathbb G=\left({\cal G}_t\right)_{t\geq 0},\ \mbox{where} \ {\cal G}_t =
\displaystyle\bigcap_{s>t}\Bigl({\cal F}_s\vee \sigma(D_u, u\leq s)\Bigr).\end{equation}
The filtration $\mathbb G$ is the new flow of information where the occurrence of $\tau$ can be detected. Mathematically, $\mathbb G$ is the smallest
right-continuous filtration that contains ${\mathbb  F }$ and makes $\tau$ being a stopping time. In the probabilistic literature, $\mathbb G$ is called the progressive enlargement of $\mathbb F$ with $\tau$. In addition to $\mathbb G$ and $D$, we associate to $\tau$ two important $\mathbb F$-supermartingales given by
\begin{equation}\label{ZandZtilde}
Z_t := ^{o,\mathbb F}(I_{\Lbrack0,\tau\Lbrack})_t=P\left(\tau >t\ \big|\ {\cal F}_t\right)\ \ \mbox{ and }\ \ \ \widetilde Z_t:=^{o,\mathbb F}(I_{\Lbrack0,\tau\Rbrack})=P\left(\tau\geq t\ \Big|\ {\cal F}_t\right).
\end{equation}
The supermartingale $Z$ (known as Az\'{e}ma supermartingale) is right-continuous with left limits, while $\widetilde Z$ admits right limits and left limits only.  The decomposition of $Z$ leads to another important martingale $m$  by
\begin{equation}\label{processmZ}
m := Z+D^{o,\mathbb F},\end{equation}
 where $D^{o,\mathbb F}$ is the $\mathbb F$-dual optional projection of $D=I_{\Lbrack\tau, \infty\Lbrack}$. Furthermore, we have $\widetilde{Z}_+ = Z$ and $\widetilde{Z} = Z_{-} + \Delta m$.\\

\noindent Throughout the paper, the filtration $\mathbb H$ denotes an arbitrary filtration satisfying the usual conditions. As usual, the set of $\mathbb H$-martingales that are $p$-integrable $(p\geq 1$) will be denoted by ${\cal M}^p(\mathbb H)$ and ${\cal A}^+(\mathbb H)$ denotes the set of increasing,
right-continuous, $\mathbb H$-adapted and integrable processes. If ${\cal C}(\mathbb H)$ is a class of processes for the filtration $\mathbb H$,
 we denote by ${\cal C}_1(\mathbb H)$ the set of processes $X\in {\cal C}(\mathbb H)$ with $X_1=0$, and
 by ${\cal C}_{loc}$
 the set  of processes $X$
 such that there exists a sequence of $\mathbb H$-stopping times, $(T_n)_{n\geq 1}$, that increases to $+\infty$ and the stopped process $X^{T_n}$ belongs
to ${\cal C}(\mathbb H)$. We put $ {\cal C}_{0,loc}={\cal C}_1\cap{{\cal
C}}_{loc}$.\\

From time to time, we need to calculate predictable projection and compensator under different filtrations. To distinguish the effect, we shall use
$^{p, \mathbb{H}} \left( V \right)$ and $\left( V\right)^{p, \mathbb{H}}$ to specify the filtration $\mathbb{H}$. When the quadratic variation process $[\is, \is]$ is $\mathbb{H}$-locally integrable, we denote   $\langle \is, \is \rangle^{\mathbb{H}}$ as its compensator.\\
For any $\mathbb H$-locally square integrable local martingale $X$, we denote $L^2_{loc}(X,\mathbb H)$ the set of predictable processes $\theta$ that are $(X,\mathbb H$)-integrable (in the semimartingale sense) and the resulting integral $\theta\is X$ is $\mathbb H$-locally square integrable local martingale.\\


\noindent Below,we recall  the notion of {\it structure conditions} that we will address in this paper. It goes back to Schweizer \cite{Schweizer1995}.

\begin{definition}\label{DefinitionofSC} Let $X$ be an $\mathbb H$-adapted process. We say that $X$ satisfies the  {\it Structure Conditions} under $\mathbb H$ (or $(X, \mathbb{H})$ satisfies SC), if there exist $M\in {\cal{M}}_{0,loc}^2(P,{\mathbb  H })$ and $\lambda \in L_{loc}^2(M, {\mathbb  H})$ such that
\begin{eqnarray}
X = X_0 + M +\lambda   \is \langle  M, M \rangle^{\mathbb H}.
\end{eqnarray}
\end{definition}

\noindent For more details about structure conditions and other related properties, we refer the reader to Schweizer \cite{Schweizer1995, Schweizer1999}, Choulli and Stricker \cite{ChoulliStricker}, and the references therein.\\



 Below, we prove a simple but useful lemma for {\it structure conditions}.

\begin{lemma}\label{lem:predictableSCNUll}
  Let $V$ be an $\mathbb H$-predictable with finite variation process. Then, $(V, \mathbb H)$ satisfies SC if and only if $V$ is constant (i.e. $V_t\equiv V_0,\ t\geq 0$).
\end{lemma}
\begin{proof}
  If $(V, \mathbb H)$ satisfies SC, then there exist an $\mathbb H$-local martingale $M^V$ and an $\mathbb H$-predictable process $\lambda^{\mathbb H} \in L^2_{loc}(M, \mathbb H)$ such that $V = V_0 + M^V + \lambda^{\mathbb H} \is\langle M^V,M^V\rangle^{\mathbb H}$. Therefore, $M^V$ is an $\mathbb H$-predictable  local martingale with finite variation. Hence $M$ is null, and $V\equiv V_0.$ This ends the proof of the lemma.\end{proof}
  The following lemma explains why when dealing with the Structure Conditions for $(S,\mathbb G)$, one can split the study into two separate cases. Precisely,  $(S,\mathbb G)$ satisfies SC is equivalent to both $(S^{\tau},\mathbb G)$ and $(S-S^{\tau},\mathbb G)$ fulfills SC.

  \begin{lemma}\label{SC2cases} The following assertions hold.\\
  {\rm{(a)}}  Let $\sigma$ be an $\mathbb H$-stopping time. Then 
 $(X,\mathbb H)$ satisfies the Structure Conditions if and only if both $(X^{\sigma},\mathbb H)$ and $(X-X^{\sigma},\mathbb H)$ do.\\
   {\rm{(b)}}    $(X,\mathbb H)$ satisfies the Structure Conditions if and only if  there exists a sequence of $\mathbb H$-stopping times that increases to infinity, $(\sigma_n)_{n\geq 1}$ such that  $(X^{\sigma_n},\mathbb H)$ satisfies the Structure Conditions for any $n\geq 1$. 
\end{lemma}
\begin{proof}
  The proof follows immediately from the definition, and will be omitted. \end{proof}

  \begin{definition}\label{GKWdecomposition}
  Let $M$ and $N$ two $\mathbb H$-local martingales with $N$ being locally square integrable. We say that $M$ admits the Galtchouk-Kunita-Watanabe decomposition, called hereafter by GKW-decomposition, with respect to $N$ if there exist a $\varphi\in L_{loc}^2(N,\mathbb H)$ and an $\mathbb H$-local martingale $L$  such that $[L, N]$ is a local martingale and 
  \begin{eqnarray*}
  M=M_0+\varphi\is N+L.\end{eqnarray*}
  \end{definition}
  
  It is well known nowadays, that this decomposition always hold when both processes $M$ and $N$ are locally square integrable local martingales. For more details about the GKW-decomposition, we refer the reader to Ansel and Stricker \cite{Ansel92, Ansel94}, and more about its applications to no-arbitrage conditions and/or other market's viability conditions we refer the reader to \cite{ChoulliStricker}.  Thus, as a direct consequence of this we state the following relationship between the SC and the  no-arbitrage notion of No-Unbounded-Profit-with-Bounded-Risk (NUPBR hereafter). NUPBR is the necessary and sufficient for the market's viability when financial agents have strictly concave utilities (the non-quadratic setting). In a series of papers, this notion of NUPBR had been deeply investigated when the filtration $\mathbb F$ is enlarged progressively with a random time. Furthermore, NUPBR and the SC notions coincide for the case when $S$ is a continuous process, while they might differ tremendously in the general case, as it is well explained in Choulli et al. \cite{ChoulliStricker}. To recall this last point, we borrow the following example from \cite{ChoulliStricker}.
  \begin{example}\label{CS96} Let $p$ be Poisson process win intensity one, $N$ is the compensated Poisson process (i.e. $N_t:=p_t-t$), and $f$ be nonnegative element of $L^1([0,1], ds)\setminus L^2([0,1], ds)$ (i.e. $\int_0^1 f(s) ds<+\infty=\int_0^1 f(s)^2 ds$).  Suppose that the filtration $\mathbb F$ is the augmented natural filtration of the Poisson process. Then the following model 
  \begin{eqnarray*}
  (S,\mathbb F),\quad\mbox{ where}\quad S:=N-\int_0^{\cdot} f(s)ds,
  \end{eqnarray*}
  satisfies NUPBR while it fails SC.
  \end{example}
 Thus, this example clearly states that NUPBR does not imply SC in general. By using the same framework and choosing $S$ to be the Poisson process itself instead, we can easily conclude also that $S$ satisfies SC but fails NUPBR. However, the relationship between NUPBR and SC can be elaborated via the GKW-decomposition as follows.
  
  \begin{theorem}\label{NUPNR/SC} Suppose that $(X,\mathbb H)$ satisfies the assumption (\ref{S-DoobMeyerDecomposition}) (i.e. $\sup_{0\leq t\leq .}X_t^2\in {\cal A}^+_{loc}(\mathbb H)$), and there exists a positive $\mathbb H$-local martingale $Z$ such that $ZX$ is a local martingale. If $Z$ admits the GKW-decomposition with respect to $M^X$  (the local martingale part of  $X$), then  $(X,\mathbb H)$  fulfills SC. \\
  In particular,  $(X, \mathbb H)$ fulfills SC as long as there exists a positive $Z\in {\cal M}_{loc}^2(\mathbb F)$ such that $ZX$ is a local martingale.
   \end{theorem}

\begin{proof} On the one hand, by applying  the GKW-decomposition of $N$ with respect to $M$, which exists by assumption, we obtain  $\lambda\in L^2_{loc}(M)$ such that  $\langle N,M\rangle=\lambda\centerdot \langle M\rangle$. On the other hand, since $X$ is a special semimartingale, having the Doob-Meyer decomposition $X+X_0+M+A$, a direct  application of It\^o formula to ${\cal E}(N)X$, we deduce that 
$$X=X_0+M-\langle N,M\rangle^{\mathbb H}=X_0+M-\lambda\centerdot \langle M\rangle.$$
 This ends the proof of the theorem.
\end{proof}


\section{Structure conditions under random horizon}\label{StructureConditions4Stau}

This section investigates and quantifies the effect of stopping at $\tau$ on the structure conditions in two different manners. On the one hand, we provide sufficient and necessary conditions on the pair $(\tau, S)$ for which the Structure Conditions hold for  $ (S^{\tau},\mathbb G)$. This answers partially the problem {\bf(Prob1)}. On the other hand, we give necessary and sufficient conditions on $\tau$ for which $ (X^{\tau},\mathbb G)$ fulfills the structure conditions as long as the model $ (X,\mathbb  F)$ does.  This section contains four subsections. The first subsection analyzes two classes of two $\mathbb G$-local martingales and their properties. The second subsection investigates the Galtchouk-Kunita-Watanabe decomposition for one of the classes of $\mathbb G$-local martingale of the first subsection. The third and fourth subsections addresses the structure conditions for $(S^{\tau},\mathbb G)$ for particular cases and general case respectively.

\subsection{Three transformation operators on $\mathbb F$-local martingales}\label{Sub3.1}

Here, we recall three transformation operators on $\mathbb F$-local martingales that play important and natural r\^oles  when investigating structure  conditions for $(S^{\tau},\mathbb G)$. Then we discuss some of their properties that will be useful throughout the rest of the section.

\begin{proposition}\label{DoobDecomposition4G}
 Let  $X\in {\cal M}_{0,loc}(\mathbb  F )$.   Then the process ${\cal T}_0(X)$ given by 
 \begin{eqnarray}\label{Mzero}
{\cal T}_0(X):=I_{\{Z_{-}>0\}}\is X-\sum \Delta X I_{\{\widetilde Z=0<Z_{-}\}}+\left(\sum \Delta X I_{\{\widetilde Z=0<Z_{-}\}}\right)^{p,\mathbb F},
\end{eqnarray}
belongs to ${\cal M}_{0,loc}(\mathbb  F )$, and the processes $\widehat{X}^{(b)}$ and ${\cal T}_b(X)$   defined by
\begin{eqnarray}\label{honestforebforeprocess}
\widehat{X}^{(b)}  &&:= X^{ \tau}  - I_{\Rbrack 0, \tau\Rbrack} Z_{-}^{-1}\is\langle X, \mm\rangle^{\mathbb F},\nonumber\\
{\cal T}_b(X)&&:=\widehat{X}^{(b)}-{1\over{\widetilde Z}}I_{\Rbrack 0, \tau\Rbrack}\is[X, m]+{1\over{Z_{-}}}I_{\Rbrack 0, \tau\Rbrack}\is\left(I_{\{\widetilde Z>0\}}\is [X,m]\right)^{p,\mathbb F},\\\nonumber
&&=X^{\tau}-{1\over{\widetilde Z}}I_{\Rbrack 0, \tau\Rbrack}\is[X, m]+I_{\Rbrack 0, \tau\Rbrack}\is\left(\sum I_{\{\widetilde Z=0<Z_{-}\}}\Delta X\right)^{p,\mathbb F},
 \end{eqnarray}
  are $\mathbb{G}$-local martingales. Furthermore, ${\cal T}_0(X)\in {\cal M}^2_{loc}(\mathbb F)$ and $\widehat{X}^{(b)}\in {\cal M}^2_{loc}(\mathbb G)$ when $X\in {\cal M}^2_{loc}(\mathbb F)$. 
\end{proposition}
\begin{proof} The proof for ${\cal T}_0(X)\in {\cal M}_{0, loc}^2(\mathbb F)$ for any $X\in  {\cal M}_{0,loc}^2(\mathbb F)$  is obvious, and will be omitted. 
The proof for $\widehat{X}^{(b)}  $  being a $\mathbb G$-local martingale  can be found in Jeulin \cite[Proposition (4.16)]{Jeu} and \cite[XX.76]{DMM}, while the proof of ${\cal T}_b(X)\in {\cal M}_{0loc}(\mathbb G)$  is given by \cite[Theorem 3]{ACJ2015}. Suppose that  $X\in {\cal M}^2_{loc}(\mathbb F)$. Then it is clear that the $\mathbb G$-predictable process with finite variation  $ I_{\Rbrack 0, \tau\Rbrack} Z_{-}^{-1}\is\langle X, \mm\rangle^{\mathbb F}$ is $\mathbb G$-locally bounded. Thus, the condition  $\widehat{X}^{(b)}  \in {\cal M}^2_{loc}(\mathbb F)$ reduces to the fact that $\displaystyle\sup_{0\leq s\leq \cdot} X_{s\wedge \tau}^2\in {\cal A}^+_{loc}(\mathbb G)$. This is always true when   $X\in {\cal M}^2_{loc}(\mathbb F)$, due to $\displaystyle\sup_{0\leq s\leq \cdot} X_{s\wedge \tau}^2\leq \sup_{0\leq s\leq \cdot} X_{s}^2 \in {\cal A}^+_{loc}(\mathbb F)\subset {\cal A}^+_{loc}(\mathbb G)$. This ends the proof of the proposition. 
 \end{proof}
The following describes the interplay between the three tranformations.
\begin{proposition}\label{prooTbXcompensators} Let $X, Y\in {\cal M}_{0,loc}(\mathbb F)$ and quasi-left-continuous.  Then the following hold.\\
{\rm (a)}  We  have 
\begin{equation}\label{XhatYbarBefore}
[\widehat X^{(b)}, {\cal T}_b(Y)]= [\widehat Y^{(b)}, {\cal T}_b(X)]={{Z_{-}}\over{\widetilde Z}}I_{\Rbrack 0, \tau\Rbrack} \is [Y, X].
\end{equation}
As a result, the process $\langle\widehat X^{(b)}, {\cal T}_b(Y)\rangle^{\mathbb G}$ always exists when $X,Y\in {\cal M}_{0,loc}^2(\mathbb F)$, and is given by 
\begin{eqnarray}\label{<x,Y>G}
\langle \widehat X^{(b)}, {\cal T}_b(Y)\rangle^{\mathbb G}=I_{\Rbrack 0, \tau\Rbrack} \is\langle {\cal T}_0(X), Y\rangle^{\mathbb F}= I_{\Rbrack 0, \tau\Rbrack} \is\langle {\cal T}_0(Y), {\cal T}_0( X)\rangle^{\mathbb F}.
\end{eqnarray}
{\rm (b)} The following equalities hold
\begin{eqnarray}
\langle \widehat{X}^{(b)},  \widehat{X}^{(b)} \rangle^{\mathbb G} &&= I_{\Rbrack 0, \tau\Rbrack}  \frac{1}{Z_{-}} \is 
\left( \widetilde{Z}\is [X, X] \right)^{\mathbb F},  \label{mmhatbefore}\\
{\cal T}_b\left({\cal T}_0(M)\right)&&= {\cal T}_b(M),\quad\mbox{and}\quad  \widehat{{\cal T}_0(M)}^{(b)}= {\widehat M}^{(b)}.\label{TbmzeroandTbm}
\end{eqnarray}
\end{proposition}

\begin{proof}   Thanks to the quasi-left-continuity of $X$ and $Y$, both processes $ I_{\Rbrack 0, \tau\Rbrack} Z_{-}^{-1}\is\langle X, \mm\rangle^{\mathbb F}$ and ${1\over{Z_{-}}}I_{\Rbrack 0, \tau\Rbrack}\is\left(I_{\{\widetilde Z>0\}}\is [Y,m]\right)^{p,\mathbb F}$ are continuous with finite variations. Thus, we derive 
\begin{eqnarray*}
[\widehat X^{(b)}, {\cal T}_b(Y)] &= & \left[X^{ \tau}  - I_{\Rbrack 0, \tau\Rbrack} Z_{-}^{-1}\is\langle X, \mm\rangle^{\mathbb F},
\widehat{Y}^{(b)}-{1\over{\widetilde Z}}I_{\Rbrack 0, \tau\Rbrack}\is[Y, m]+{1\over{Z_{-}}}I_{\Rbrack 0, \tau\Rbrack}\is\left(I_{\{\widetilde Z>0\}}\is [Y,m]\right)^{p,\mathbb F}
\right]\\
&= & \left[X^{ \tau} , 
 Y^{ \tau} - {1\over{\widetilde Z}}I_{\Rbrack 0, \tau\Rbrack}\is[Y, m]
\right]={{Z_{-}}\over{\widetilde Z}}I_{\Rbrack 0, \tau\Rbrack} \is [Y, X].
\end{eqnarray*}
The symmetric role of $X$ and $Y$ leads to $ [\widehat Y^{(b)}, {\cal T}_b(X)]={{Z_{-}}\over{\widetilde Z}}I_{\Rbrack 0, \tau\Rbrack} \is [Y, X]$. 
This proves (\ref{XhatYbarBefore}). To prove the second statement of assertion (a), we combine equation (\ref{XhatYbarBefore}) and Lemma \ref{lemmecrucial} and derive
\begin{eqnarray*}
\langle \widehat X^{(b)}, {\cal T}_b(Y) \rangle^{\mathbb G} &=& \left( {{Z_{-}}\over{\widetilde Z}}I_{\Rbrack 0, \tau\Rbrack} \is [X, Y] \right)^{p,\mathbb G}  = I_{\Rbrack 0, \tau\Rbrack} \is \left( I_{\{\widetilde{Z} >0 \}} \is [X,Y] \right)^{p,\mathbb F}\\
&=& I_{\Rbrack 0, \tau\Rbrack} \is \langle {\cal T}_0(X), Y\rangle^{\mathbb F}= I_{\Rbrack 0, \tau\Rbrack} \is \langle {\cal T}_0(X),  {\cal T}_0(Y)\rangle^{\mathbb F}.
\end{eqnarray*}
This ends the proof of assertion (a). Thanks again to the quasi-left-continuity of $X$, we have 
\begin{eqnarray*}
\left[ \widehat{X}^{(b)}, \widehat{X}^{(b)}\right] &=& 
\left[X^{ \tau}  - I_{\Rbrack 0, \tau\Rbrack} Z_{-}^{-1}\is\langle X, \mm\rangle^{\mathbb F}, X^{ \tau}  - I_{\Rbrack 0, \tau\Rbrack} Z_{-}^{-1}\is\langle X, \mm\rangle^{\mathbb F}\right] \\
&=& I_{\Rbrack 0, \tau\Rbrack}\is [X,X].
\end{eqnarray*}
Thus, from combining this with Lemma \ref{lemmecrucial}, (\ref{mmhatbefore}) follows immediately, while (\ref{TbmzeroandTbm}) is a direct consequence of $I_{\{\widetilde Z=0\}}I_{\Lbrack0,\tau\Rbrack}\equiv 0$. This ends the proof of the proposition.\end{proof}

The following proposition connects the integrability with respect to $({\widehat X}^{(b)}, \mathbb G)$ to that with respect to $({\cal T}_0(X),\mathbb F)$, for any quasi-left-continuous $X\in {\cal M}_{loc}(\mathbb F)$. 
\begin{proposition}\label{propositionmzerotaob}  Let $M\in {\cal M}_{0,loc}^2(\mathbb  F)$ be given by (\ref{S-DoobMeyerDecomposition}), and suppose that $M$ is quasi-left-continuous.Then there exists a unique ($P\otimes d\langle {\cal T}_0(M)\rangle^{\mathbb F}$-a.e.) $\mathbb F$-predictable process $\widetilde\varphi$ such that
\begin{eqnarray}\label{processVarphi}
I_{\{Z_{-}>0\}}\is\left(\widetilde Z\is [M,M]\right)^{p,\mathbb F}=\widetilde\varphi\is \langle {\cal T}_0(M)\rangle^{\mathbb F},\quad \mbox{and}\quad0\leq \widetilde\varphi \leq 1. \end{eqnarray}
Furthermore, the following assertions hold.\\
{\rm{(a)}} It holds that  $I_{\{\widetilde\varphi=0\}}\is{\widehat M}^{(b)}\equiv 0$, and  $\displaystyle{{\theta}\over{\sqrt{\widetilde\varphi}}}I_{\{\widetilde\varphi>0\}}\in L^2_{loc}({\widehat M}^{(b)}, \mathbb G)$, for any $\theta\in  L^2_{loc}(M, \mathbb F).$\\
{\rm{(b)}}  Let $\theta^{\mathbb G}$ be a $\mathbb G$-predictable process. Then $\theta^{\mathbb G}\in L^2_{loc}({\widehat M}^{(b)}, \mathbb G)$ if and only if there exists and $\mathbb F$-predictable process $\theta$ such that $\theta^{\mathbb G}=\theta$ on $\Rbrack0,\tau\Rbrack$, and 
$\theta\sqrt{\widetilde\varphi}I_{\{Z_{-}\geq \delta\}}\in L^2_{loc}({\cal T}_0(M),\mathbb F)$ for any $\delta>0$. 
\end{proposition}
\begin{proof} 
To prove the existence of $\widetilde\varphi$ satisfying (\ref{processVarphi}), we consider
\begin{eqnarray*}
N:=Z_{-}\is M+I_{\{Z_{-}>0\}}\is [M,m]-I_{\{Z_{-}>0\}}\is \langle M, m\rangle^{\mathbb F},
\end{eqnarray*}
which is a locally square integrable and quasi-left-continuous  $\mathbb F$-local martingale and its jumps are given by $\Delta N = \widetilde{Z} \Delta M $. Thus, a direct application of the GKW-decomposition leads to the existence of a unique pair $(\widetilde\varphi, L)\in L^2_{loc}( {\cal T}_0(M),\mathbb F)\times {\cal M}_{0,loc}^2(\mathbb F)$ such that 
\begin{eqnarray*}
N={\widetilde\varphi}\is {\cal T}_0(M)+L,\quad \langle {\cal T}_0(M), L\rangle^{\mathbb F}\equiv 0.\end{eqnarray*}
From the definition of $N$, we easily derive  
\begin{eqnarray*}
[N,{\cal T}_0(M)] =  [N,M]= \widetilde{Z} I_{\{Z_{-}>0\}}\is [M, M].
\end{eqnarray*}
Thus, by compensating above on both sides, the first equality in (\ref{processVarphi}) follows immediately. Then due to the fact that  $\widetilde Z I_{\{Z_{-}>0\}}\is [M,M]$ is nondecreasing and is dominated by $[{\cal T}_0(M),{\cal T}_0(M)]$, we deduce that $0\leq\widetilde\varphi\leq 1$, and (\ref{processVarphi}) is proved. \\
(a) To prove assertion (a), on the one hand, we remark that 
\begin{eqnarray*}
\langle I_{\{\widetilde\varphi=0\}}\is{\widehat M}^{(b)}\rangle^{\mathbb G}=I_{\{\widetilde\varphi=0\}}\is\langle {\widehat M}^{(b)}\rangle^{\mathbb G}={{{\widetilde\varphi}}\over{Z_{-}}}I_{\{\widetilde\varphi=0\}}I_{\Rbrack0,\tau\Rbrack}\is\langle {\cal T}_0( M)\rangle^{\mathbb F}\equiv 0.
\end{eqnarray*}
On the other hand, we calculate 
\begin{eqnarray*}
\langle{{\theta}\over{\sqrt{\widetilde\varphi}}}I_{\{\widetilde\varphi>0\}}\is {\widehat M}^{(b)}\rangle^{\mathbb G}={{\theta^2}\over{\widetilde\varphi}}I_{\{\widetilde\varphi>0\}}\is\langle {\widehat M}^{(b)}\rangle^{\mathbb G}={{\theta^2}\over{Z_{-}}}I_{\{\widetilde\varphi>0\}}I_{\Rbrack0,\tau\Rbrack}\is\langle {\cal T}_0( M)\rangle^{\mathbb F}.\end{eqnarray*}
Therefore, assertion (a) follows immediately from these equalities,  . \\
(b) Notice the equality
\begin{eqnarray*}
(\theta^{\mathbb G})^2 \is \langle {\widehat M}^{(b)} \rangle^{\mathbb G} &=& 
\frac{\theta^2}{Z_{-} } I_{\Rbrack0,\tau\Rbrack}\is \left(\widetilde Z\is [M,M]\right)^{p,\mathbb F}
=\frac{(\theta\sqrt{\widetilde\varphi})^2}{Z_{-} } I_{\Rbrack 0,\tau\Rbrack}\is \langle  {\cal T}_0( M) \rangle^{\mathbb F}.
\end{eqnarray*}
Then assertion (b) follows from combining this equality with and  Proposition \ref{G/Flocalization}-(c).
This ends the proof of the proposition.
 \end{proof} 
\subsection{Galtchouk-Kunita-Watanabe for a class of $\mathbb G$-local martingales}\label{Sub3.2}
This subsection investigates the GKW-decomposition of ${\cal T}_b(M)$ with respect to ${\widehat M}^{(b)}$, when $M$ belongs to $ {\cal M}_{0,loc}^2(\mathbb F)$ . In fact, we prove that this decomposition always holds for $ \sqrt{\widetilde\varphi}\is{\cal T}_b(M)$ instead of ${\cal T}_b(M)$.

\begin{theorem}\label{GKW4G}  Let $M\in {\cal M}_{0,loc}^2(\mathbb  F)$ be given by (\ref{S-DoobMeyerDecomposition}), and ${\cal T}_0(M)$ be given by (\ref{Mzero}). Suppose that $M$ is quasi-left-continuous, Then the following assertions hold.\\
 {\rm{(a)}}  There exists $\widetilde L\in {\cal M}_{0,loc}(\mathbb G)$ such that $[\widetilde L, {\widehat M}^{(b)}]\in {\cal M}_{0,loc}(\mathbb G)$  and 
 \begin{eqnarray}\label{GKW4Tb(M)}
 \sqrt{\widetilde\varphi}\is{\cal T}_b(M)={{Z_{-}}\over{\sqrt{\widetilde\varphi}}} I_{\{\widetilde\varphi>0\}}\is {\widehat M}^{(b)}+\sqrt{\widetilde\varphi}\is\widetilde{L}.
 \end{eqnarray}
 {\rm{(b)}}  For any quasi-left-continuous $N\in {\cal M}_{0,loc}^2(\mathbb F)$, there exist unique pair $(\theta_N, L_N)$ that belongs to  $L^2_{loc}({\cal T}_0(M),\mathbb F)\times {\cal M}_{0,loc}(\mathbb G)$ and satisfies  
 \begin{eqnarray}\label{GKW4Tb(N)}
\sqrt{\widetilde\varphi}\is {\cal T}_b(N)={{Z_{-}\theta_N}\over{\sqrt{\widetilde\varphi}}}I_{\{\widetilde\varphi>0\}}\is {\widehat M}^{(b)}+\sqrt{\widetilde\varphi}\is L_N,\quad [L_N, {\widehat M}^{(b)}]\in {\cal M}_{0,loc}(\mathbb G),
 \end{eqnarray}
 with 
 \begin{eqnarray}
 \theta_N :=  \frac{{d\langle N, {\cal T}_0(M)}\rangle^{\mathbb F}}{d\langle {\cal T}_0(M)\rangle^{\mathbb F}}.
 \end{eqnarray}
\end{theorem}


\begin{proof}
For any locally square integrable    $\mathbb F$-local martingale $N$,  denote  ${\cal T}_{(b,n)}(N)$ as 
\begin{eqnarray}\label{Nn}
{\cal T}_{(b,n)}(N):=\widehat{N}^{(b)} -{I_{\Rbrack 0, \tau\Rbrack}\over{Z_{-}}}\is{\cal T}_b(N_n), \mbox{where} \ N_n:=I_{\{\widetilde Z\geq n^{-1}\}}\is[N, m]-
\left(I_{\{\widetilde Z\geq n^{-1}\}}\is [N,m]\right)^{p,\mathbb F}.
\end{eqnarray} 
   It is clear that ${\cal T}_{(b,n)}(N)$ converges in ${\cal M}_{loc}(\mathbb G)$ to ${\cal T}_{b}(N)$, $N_n$ converges in  ${\cal M}_{loc}^2(\mathbb G)$, and $\langle  {\cal T}_{(b,n)}(N) , \widehat M^{(b)}\rangle^{\mathbb G}$ converges to  $\langle  {\cal T}_{b}(N) , \widehat M^{(b)}\rangle^{\mathbb G}$. Furthermore, a direct application of GKW-decomposition of ${\cal T}_{(b,n)}(N)$ with respect to ${\widetilde\varphi}^{-1/2}\is {\widehat M}^{(b)}$, since both processes belong to ${\cal M}_{loc}^2(\mathbb G)$, leads to  
 \begin{eqnarray}\label{GKW1}
 {\cal T}_{(b,n)}(N)={{\theta_n}\over{\sqrt{\widetilde\varphi}}}\is {\widehat M}^{(b)}+L_n.\end{eqnarray}
 Here $\theta_n$ is an $\mathbb F$-predictable process  such that $\theta_n\in L^2_{loc}({\widetilde\varphi}^{-1/2}\is {\widehat M}^{(b)}, \mathbb G)$ and $L_n\in {\cal M}^2_{loc}(\mathbb G)$  with  $\langle L_n, {\widehat M}^{(b)}\rangle^{\mathbb G}\equiv 0$.  As a result, see Proposition \ref{propositionmzerotaob}-(b), we conclude that $\theta_n I_{\{ Z_{-}\geq\delta\}}\in  L^2_{loc}( {\cal T}_0(M), \mathbb F)$, for any $\delta\in (0,1)$. \\
Now, we focus on proving that $\theta_n $ indeed converges to $\widetilde\theta$ in  $L^2_{loc}({\widehat M}^{(b)}, \mathbb G)$, or equivalently $\theta_n\is {\widehat M}^{(b)}$ converges in the space ${\cal M}_{0,loc}^2(\mathbb G)$. To this end, we consider an $\mathbb F$-predictable process $\eta \in L^2_{loc}({\widehat M}^{(b)}, \mathbb G)$, which is equivalent to $\eta \sqrt{\widetilde\varphi} I_{\{ Z_{-}\geq\delta\}}\in L^2_{loc}({\cal T}_0(M),\mathbb F)$ for any $\delta\in (0,1)$, and  derive
\begin{eqnarray*}
(\theta_{n+k}-\theta_n) \eta \is\langle {\widehat M}^{(b)}  \rangle^{\mathbb G}&=&{{\theta_{n+k}-\theta_n}\over{\sqrt{\widetilde\varphi}}}\eta \sqrt{\widetilde\varphi}\is \langle {\widehat M}^{(b)} \rangle^{\mathbb G} \\
&=&\langle {{\theta_{n+k}-\theta_n}\over{\sqrt{\widetilde\varphi}}}I_{\{\widetilde\varphi>0\}}\is{\widehat M}^{(b)}, \eta \sqrt{\widetilde\varphi}\is{\widehat M}^{(b)}\rangle^{\mathbb G} \\
&=&  \langle {\cal T}_{(b,n+k)}(N)-{\cal T}_{(b,n)}(N), \widehat{\eta \sqrt{\widetilde\varphi}\is M}\rangle^{\mathbb G}  \\
&=&  \langle {\cal T}_{b}(N_{n+k}-N_{n}), \widehat{\eta \sqrt{\widetilde\varphi}\is{\cal T}_0( M)}\rangle^{\mathbb G}\\
&=&  I_{\Rbrack 0, \tau\Rbrack}\is \langle N_{n+k}-N_{n}, \eta \sqrt{\widetilde\varphi}\is{\cal T}_0( M)\rangle^{\mathbb F} .
\end{eqnarray*}
Since  $\eta \sqrt{\widetilde\varphi}\is {\cal T}_0( M)\in {\cal M}_{0,loc}^2(\mathbb F)$ and the sequence $N_n$ converges  in ${\cal M}_{0,loc}^2(\mathbb F)$, we deduce that $(\theta_n)_n$ is a Cauchy sequence in $L^2_{loc}({\widehat M}^{(b)}, \mathbb G)$, and hence there exists a unique $\widetilde\theta\in L^2_{loc}({\widehat M}^{(b)}, \mathbb G)$ such that $\theta_n$ converges to $\widetilde\theta$ in $L^2_{loc}({\widehat M}^{(b)}, \mathbb G)$. 
Since  ${\cal T}_{(b,n)}(N)$ converges to ${\cal T}_{b}(N)$ in the space ${\cal M}_{0,loc}(\mathbb G)$, we deduce that $L_n$ converges in ${\cal M}_{0,loc}(\mathbb G)$ to $L\in {\cal M}_{0,loc}(\mathbb G)$. It is easy to check that $[L, {\widehat M}^{(b)} ]\in {\cal M}_{0,loc}(\mathbb G)$, and 
\begin{eqnarray}\label{GKW1bis}
\sqrt{\widetilde\varphi}\is {\cal T}_{b}(N)={\widetilde\theta}\is{\widehat M}^{(b)}+\sqrt{\widetilde\varphi}\is L.
\end{eqnarray} 
Then by considering $N=M$ and using Propositions \ref{prooTbXcompensators} and  \ref{propositionmzerotaob}, we derive 
\begin{eqnarray*}
{\widetilde\theta}{{\widetilde\varphi}\over{Z_{-}}}  I_{\Rbrack 0, \tau\Rbrack}\is\langle {\cal T}_0(M)\rangle^{\mathbb F}={\widetilde\theta}\is\langle {\widehat M}^{(b)}\rangle^{\mathbb G}= \langle \sqrt{\widetilde\varphi}\is {\cal T}_{b}(M), {\widehat M}^{(b)} \rangle^{\mathbb G}= \sqrt{\widetilde\varphi}\is \langle{\cal T}_{b}(M), {\widehat M}^{(b)} \rangle^{\mathbb G}= \sqrt{\widetilde\varphi} I_{\Rbrack 0, \tau\Rbrack}\is\langle {\cal T}_0(M)\rangle^{\mathbb F}.
\end{eqnarray*}
This proves that $\widetilde\theta$ coincides with $Z_{-}/\sqrt{\widetilde\varphi}$ on $\Rbrack 0, \tau\Rbrack$, and hence (\ref{GKW1bis}) for $N=M$ becomes 
\begin{eqnarray*}
\sqrt{\widetilde\varphi}\is {\cal T}_{b}(N)={{Z_{-}}\over{\sqrt{\widetilde\varphi}}}\is{\widehat M}^{(b)}+\sqrt{\widetilde\varphi}\is L.
\end{eqnarray*}
This proves assertion (a). To prove assertion (b), we consider $N\in {\cal M}_{loc}^2(\mathbb F)$, and we apply the GKW-decomposition  for $N$ with respect to ${\cal T}_0(M)$. This implies the existence of unique pair $(\theta_N, L_0)$ that belongs to $ L^2_{loc}( {\cal T}_0(M),\mathbb F)\times {\cal M}_{0,loc}^2(\mathbb F)$ and satisfies
\begin{eqnarray*}
N=\theta_N\is {\cal T}_0(M)+L_0,\quad \langle {\cal T}_0(M),L_0\rangle^{\mathbb F}\equiv 0.
\end{eqnarray*}
Remark that  we have ${\cal T}_b(N)= \theta_N\is {\cal T}_b({\cal T}_0(M))+{\cal T}_b(L_0)$, ${\cal T}_b({\cal T}_0(M))={\cal T}_b(M)$, and 
\begin{eqnarray*}
\langle {\cal T}_b(L_0), {\widehat M}^{(b)} \rangle^{\mathbb G}= I_{\Rbrack 0, \tau\Rbrack}\is\langle L_0, {\cal T}_0(M)\rangle^{\mathbb F}\equiv 0.
\end{eqnarray*}
Thus, by putting $L_N:={\cal T}_b(L_0)+\theta_N\is {\widetilde L}$, we deduce that $\langle L_N, {\widehat M}^{(b)}\rangle^{\mathbb G}\equiv 0$ and (\ref{GKW4Tb(N)}) follows immediately from combining these with assertion (a). This ends the proof of the theorem.
\end{proof}

\begin{corollary}  Let $M$ be a quasi-left-continuous element of ${\cal M}_{0,loc}(\mathbb F)$ such that $\Delta M\Delta m\equiv 0$. Then ${\cal T}_b(M)=\widehat M^{(b)}$.
\end{corollary}
The proof of the corollary is immediate due to the fact that in this case $\widetilde\varphi=Z_{-}$.


\subsection{Particular cases and examples for the model $(S^{\tau},\mathbb G)$}

We start this subsection by proving that the {\it structure conditions} hold whenever $S$ or $m$ is continuous.

\begin{theorem}\label{continuouscaseBefore} If $(S,\mathbb F)$ satisfies the SC, with its market's price $\widehat\lambda^{\mathbb F}$, and either $S$ or $m$ is continuous, then the model $(S^{\tau},\mathbb G)$ fulfills the SC, and its market's price of risk $\widehat\lambda^{\mathbb G}$ is given by 
\begin{eqnarray}\label{MarketPriceofRiskContinuousBefore}
{\widehat\lambda}^{\mathbb G} := \left( {\widehat\lambda}^{\mathbb F}+\frac{\beta^{(m)}}{Z_{-}} \right) I_{\Rbrack 0,\tau \Rbrack},\quad\mbox{with}\quad \beta^{(m)}:={{d\langle m,M^S\rangle^{\mathbb F}}\over{d\langle M^S\rangle^{\mathbb F}}}.
\end{eqnarray}
 \end{theorem}
\begin{proof} 
Suppose $(S,\mathbb F)$ satisfies SC with the market price $\widehat\lambda^{\mathbb F}$ and   $S = S_0 + M^S +{\widehat\lambda}^{\mathbb F} \is \langle M^S, M^S\rangle^{\mathbb F}$, where $M\in {\cal M}^2_{0,loc}(\mathbb{F})$ and ${\widehat\lambda}^{\mathbb F} \in L^2_{oc}(M,\mathbb F)$. For notational simplicity, we put $M:= M^S$ and  
 \begin{eqnarray}
 \widehat{M} : = I_{\Rbrack 0,\tau \Rbrack}\is M - \left({Z_{-}}\right)^{-1}I_{\Rbrack 0,\tau \Rbrack}\is \langle M,m\rangle^{\mathbb F},
 \end{eqnarray}
 which is a $\mathbb{G}$-locally square integrable local martingale. Then the Doob-Meyer decomposition of $S^\tau$ under $\mathbb{G}$ is given by 
 \begin{eqnarray}
 S^\tau = S_0 + \widehat{M}  + {\widehat\lambda}^{\mathbb F} I_{\Rbrack 0,\tau \Rbrack}\is \langle M,M\rangle^{\mathbb F} +\left({Z_{-}}\right)^{-1}I_{\Rbrack 0,\tau \Rbrack}\is \langle M,m\rangle^{\mathbb F}.
 \end{eqnarray}
Thus, the proof will follow as long as we find a $\mathbb{G}$-predictable process $\widehat{\lambda}^{\mathbb G} \in L^2_{loc}(\widehat{M},\mathbb G)$ such that
 \begin{eqnarray}
 {\widehat\lambda}^{\mathbb F} I_{\Rbrack 0,\tau \Rbrack}\is \langle M,M\rangle^{\mathbb F} + \left({Z_{-}}\right)^{-1}I_{\Rbrack 0,\tau \Rbrack}\is \langle M,m\rangle^{\mathbb F} = \widehat{\lambda}^{\mathbb G}\is \langle \widehat{M}\rangle^{\mathbb G}.
 \end{eqnarray}
 To this end, since $m$ is locally bounded,   the GKW-decomposition of $m$ with respect to $M$ (under $\mathbb F$)    implies the existence of an $\mbf$-predictable process ${{\beta^{(m)}}}\in L^2_{loc}(M,\mathbb F)$ and a locally square integrable $\mathbb{F}$-local martingale $m^\perp$ such that
 \begin{eqnarray}
 m = m_0 + {\beta^{(m)}}\is M + m^\perp\ \ \ \ \  \mbox{ and  }\ \ \ \  \ \langle M,  m^\perp \rangle^{\mathbb{F}} =0.
 \end{eqnarray}
 Therefore, 
 \begin{eqnarray}\label{drift4G}
 {\widehat\lambda}^{\mathbb F} I_{\Rbrack 0,\tau \Rbrack}\is \langle M,M\rangle^{\mathbb F} +\frac{1}{Z_{-}}I_{\Rbrack 0,\tau \Rbrack}\is \langle M,m\rangle^{\mathbb F} 
 &=& \left({\widehat\lambda}^{\mathbb F} +\frac{{\beta^{(m)}}}{Z_{-}} \right) I_{\Rbrack 0,\tau \Rbrack}\is \langle M,M\rangle^{\mathbb F}.
 \end{eqnarray}
The continuity of $S$ or $m$ leads to 
   $$\Delta \langle M, m\rangle^{\mathbb F}=0,\quad \widetilde{Z}\is [M]=Z_{-}\is [M],
   $$
and
   \begin{eqnarray*}\langle \widehat{M}\rangle^{\mathbb G}&=&\left([\widehat M]\right)^{p,\mathbb G}= \left([ M]^{\tau}\right)^{p,\mathbb G}={1\over{Z_{-}}} I_{\Rbrack 0,\tau \Rbrack} \is \left({\widetilde Z}\is [M]\right)^{p,\mathbb F}= I_{\Rbrack 0,\tau \Rbrack} \is\langle M\rangle^{\mathbb F}.
   \end{eqnarray*}
   
   By inserting above equalities  in (\ref{drift4G}), we get
 \begin{eqnarray}
 S^\tau = S_0 + \widehat{M} + {\widehat\lambda}^{\mathbb G}  \is \langle \widehat{M}\rangle^\mbg,\ \ \ \ \mbox{with the market price of risk  }\ \ {\widehat\lambda}^{\mathbb G} := \left( {\widehat\lambda}^{\mathbb F}+\frac{{\beta^{(m)}}}{Z_{-}} \right) I_{\Rbrack 0,\tau \Rbrack}.
 \end{eqnarray}
  It is clear that $\widehat{\lambda}^{\mathbb G} \in L^2_{loc}(\widehat{M}, \mathbb G)$ due to the local boundedness of $\left({Z_{-}}\right)^{-1}I_{\Rbrack 0,\tau\Rbrack}$. This ends the proof of theorem.\end{proof}

\begin{corollary}
Suppose that $S$ is a local martingale (i.e. $A^S\equiv 0$), and either $S$ is a continuous or $m$ is continuous, then $(S^{\tau},\mathbb G)$ satisfies SC, and its market's price of risk is $\beta^{(m)}Z_{-}^{-1}I_{\Rbrack 0,\tau \Rbrack}$.
  \end{corollary}
  
 The rest of this subsection is devoted to show, via simple examples, that the SC might be violated for some random times. 

\begin{proposition} Suppose that the stochastic basis $(\Omega, \mathbb{A},\mathbb{F}=({\cal F}_t)_{t\geq 0}, P)$ supports  a Poisson process $N$ with intensity $\lambda$, and the stock price --denoted by $X$-- is given by
$$
  dX_t = X_{t-} \psi dK_t, \ \ \mbox{where}, \ \ \psi >-1, \ \mbox{and} \ \psi \neq 0,\ \ \ \  K_t=N_t - \lambda t.$$
  Then the following assertions hold.\\  
  {\rm{(a)}} If
  \begin{eqnarray*}
  \tau = \alpha T_1+(1-\alpha)T_2,\ \ \ \ \ \mbox{where}\ \ \ \ T_i &=& \inf\{t\geq 0: N_t\geq i\}, \ i\geq 1,\ \ \ \alpha\in(0,1),
   \end{eqnarray*}
then $X^\tau$ does not satisfy SC$(\mathbb{G})$.\\
{\rm{(b)}} If $\tau =(\alpha T_2)\wedge T_1$, then $(Y^{\tau}, \mathbb G)$ satisfies SC, where $Y:=K$.
\end{proposition}

\begin{proof}a) Here we prove assertion (a). To this end, we recall from Aksamit et al. \cite{aksamit/choulli/deng/jeanblanc} that the Az\'ema supermartingale  $Z$ and $m$ take the forms of:
\begin{equation}
 Z = I_{\Lbrack 0, T_1 \Lbrack} + \phi^m I_{\Lbrack T_1, T_2 \Lbrack}, \ \ m= 1 - \phi^m I_{\Rbrack T_1, T_2 \Rbrack}\is K, \ \ \mbox{where} \  \phi^m_t = e^{-\lambda \frac{k_0}{k_2}(t - T_1)}.
\end{equation}
 Then, it is easy to calculate  that
 \begin{eqnarray*}
  \frac{1}{Z_{-}}I_{\Lbrack 0,\tau \Rbrack} \is \left\langle X, m\right\rangle^{\mathbb F}_t &=&   \frac{-1}{Z_{-}}I_{\Rbrack T_1,\tau \Rbrack}X_{-}\psi \phi^m \is \left \langle K\right\rangle^{\mathbb F}_t  \nonumber \\
  &=&  - \lambda\int_0^t    \frac{1}{Z_{u-}}I_{\Rbrack T_1,\tau \Rbrack}X_{u-}\psi_u \phi^m_u  du = - \lambda\int_0^t      X_{u-}\psi_u   I_{\Rbrack T_1, \tau \Rbrack} du
  \end{eqnarray*}
  and
  \begin{eqnarray*}
  \left \langle \widehat{X}^{(b)}, \widehat{X}^{(b)} \right\rangle^{\mathbb{G}}_t&=&   I_{\Lbrack 0,\tau \Rbrack}\is \left \langle X, X \right\rangle^{\mathbb F}_t + \frac{1}{Z_{-}}I_{\Lbrack 0,\tau \Rbrack}\is \left( \sum \Delta m (\Delta X)^2 \right)^{p,\mathbb{F}}_t \\
  &=& \lambda \int_0^t X_{u-}^2 \psi^2_u I_{\Lbrack 0,\tau \Rbrack} du  - \lambda\int_0^t \frac{1}{Z_{u-}}X_{u-}^2\psi_u^2 \phi^m_u  I_{\Rbrack T_1,\tau \Rbrack}du \\
  &=&  \lambda\int_0^t X_{u-}^2\psi_u^2  I_{\Lbrack 0, T_1 \Rbrack}du,
  \end{eqnarray*}
  where $\widehat{X}^{(b)}$ is defined via (\ref{honestforebforeprocess}).   Hence, there is no $\mathbb{G}$-predictable process $\widehat{\lambda} \in L^2_{loc}(\widehat{X}^{(b)})$ satisfying
  \begin{equation}
    \frac{1}{Z_{-}}I_{\Lbrack 0,\tau \Rbrack} \is \left\langle X, m\right\rangle^{\mathbb F} = \widehat{\lambda}\is   \left \langle \widehat{X}^{(b)}, \widehat{X}^{(b)} \right\rangle^{\mathbb{G}},
  \end{equation}
  since $\Rbrack T_1, \tau \Rbrack$ and $\Lbrack 0, T_1 \Rbrack$ are disjoint. This proves assertion (a).\\
  b) This part proves assertion (b). Thanks to the calculations in \cite[Example 2.12]{ACDM}, in this case of $\tau$, we have 
\begin{eqnarray*}
Y_{t\wedge\tau}=N_{t\wedge\tau}-\tau\wedge
t,\quad {\widehat Y}^{(b)}_t=Y_{t\wedge\tau}-\int_0^{\tau\wedge t} {{\beta u}\over{\beta
u+1}}du,\quad  \beta:=\alpha^{-1}-1>0,\end{eqnarray*}
and $Y^{\tau} {\cal E}(\theta\is {\widehat Y}^{(b)})$ are $\mathbb G$-local martingales, where ${\cal E}(\theta\is {\widehat Y}^{(b)})$ is the stochastic exponential of $\theta\is \widehat Y$ and $\theta_u:=(1+\beta u)/(1+2\beta u)-1=-(\beta u)/(1+2\beta u)$. Hence, since $\theta$ is locally bounded, we deduce that ${\cal E}(\theta\is{\widehat Y}^{(b)})\in{\cal M}_{loc}^2(\mathbb G)$. Thus, a direct application of Theorem \ref{NUPNR/SC} leads to conclude that $Y^{\tau}$ fulfills SC. This ends the proof of the proposition.
 \end{proof}

\subsection{The general case for  $(S^{\tau},\mathbb G)$}

The following is  the first main result of this section, where we treat fully  the problem {\bf(Prob1)} for the case when $S$ is quasi-left-continuous.
\begin{theorem}\label{generlaQLC}
 Suppose $S$ is quasi-left-continuous. Then the following are equivalent.\\
{\rm{ (a)}} $(S^{\tau}, \mathbb G)$ satisfies SC.\\
 {\rm{ (b)}} For any $\delta>0$, the model $(I_{\{Z_{-}\geq \delta\}}\is S^{(0)},\mathbb F)$ satisfies SC, and its market's price of risk $\lambda^{(0,\mathbb F)}$ satisfies $(\lambda^{(0,\mathbb F)}Z_{-}+{\beta^{(0,m)}})( \widetilde\varphi)^{-1/2}I_{\{Z_{-}\geq\delta,\  \widetilde\varphi>0\}}\in L^2_{loc}({\cal T}_0(M),\mathbb F)$, where
 \begin{equation}\label{process Szero}
 S^{(0)}:= S-\sum I_{\{ \widetilde Z=0<Z_{-}\}}\Delta M^S=:S_0+{\cal T}_0(M)+A^{(0)},\quad 
 A^{(0)}:=A-\left(\sum I_{\{ \widetilde Z=0<Z_{-}\}}\Delta M^S\right)^{p,\mathbb F}.
 \end{equation}
 
 Furthermore, the market's prices of risk $\widetilde{\lambda}^{\mathbb G}$ and $\widetilde{\lambda}^{(0,\mathbb F)}$, for $(S^{\tau}, \mathbb G)$  and $(I_{\{Z_{-}\geq \delta\}}\is S^{(0)},\mathbb F)$ respectively, are related by the following
 \begin{eqnarray}
 &&\widetilde{\lambda}^{\mathbb G}= {{Z_{-}\widetilde{\lambda}^{(0,\mathbb F)}+{\beta^{(0,m)}} }\over{ \widetilde{\varphi}}}I_{\{ \widetilde\varphi>0\}} I_{\Rbrack0,\tau\Rbrack}\quad\mbox{and}\quad Z_{-}^2\widetilde{\lambda}^{(0,\mathbb F)}=\ ^{p,\mathbb F}(\widetilde{\lambda}^{\mathbb G}I_{\Rbrack0,\tau\Rbrack}) \widetilde{\varphi}- {\beta^{(0,m)}} Z_{-}, \label{G2F}\\
 &&\mbox{where}
 \quad {\beta^{(0,m)}}:={{d\langle m, {\cal T}_0(M) \rangle^{\mathbb F}}\over{d\langle {\cal T}_0(M)\rangle^{\mathbb F}}},\quad\mbox{and}\quad  \widetilde\varphi:={{d\left(\widetilde Z\is [M,M]\right)^{p,\mathbb F}}\over{d \langle {\cal T}_0(M)\rangle^{\mathbb F}}}.\label{F2G}
  \end{eqnarray}
 \end{theorem}

\begin{proof}
The proof of this theorem is achieved in two steps. The first step proves (b)$\Longrightarrow$ (a), while the second step proves the converse. \\
{\bf Step 1.} Here we prove (b)$\Longrightarrow$ (a). To this end, we assume that assertion (b) holds. Then the quasi-left-continuity of $S$ implies the quasi-left-continuity of $M:=M^S$. We start by recalling some useful equalities from Proposition \ref{prooTbXcompensators} and (\ref{TbmzeroandTbm}) in Proposition \ref{propositionmzerotaob} as follows.
\begin{eqnarray*}\label{mM/MMquasileft}
\langle{\cal T}_b({M}),\widehat{M}^{(b)} \rangle^{\mathbb G}&=&I_{\Rbrack0,\tau\Rbrack}\is\langle {\cal T}_0(M)\rangle^{\mathbb F},\ \mbox{and}  \
 \langle{\cal T}_b(m),\widehat{M}^{(b)} \rangle^{\mathbb G} = I_{\Rbrack0,\tau\Rbrack}\is  \langle m,{\cal T}_0(M)\rangle^{\mathbb F} ,\\
 {\cal T}_b({\cal T}_0(M))&=& {\cal T}_b(M),\quad\mbox{and}\quad  \widehat {{\cal T}_0(M)}^{(b)}  = {\widehat M}^{(b)} .
\end{eqnarray*}
Let $\tau_{\delta}$ be $\mathbb G$-stopping time such that $\Rbrack0,\tau\wedge\tau_{\delta}\Rbrack\subset \{Z_{-}\geq \delta\}$, which is possible since $Z_{-}^{-1}I_{\Rbrack0,\tau}$ is $\mathbb G$-locally bounded. 
Since  $S^{\tau}=(S^{(0)})^{\tau}$ and $S^{(0)}$ satisfies SC,  we deduce the existence  of an $\mathbb F$-predictable process $\widetilde{\lambda}^{(0,\mathbb F)}$ such that $A^{(0)}=\widetilde{\lambda}^{(0,\mathbb F)}\is \langle {\cal T}_0(M)\rangle^{\mathbb F}$. Hence, by combining all these with Theorem \ref{GKW4G}, we derive
\begin{eqnarray*}
S^{\tau}&&=S_0+({\cal T}_0(M))^{\tau}+(A^{(0)})^{\tau}=S_0+{\widehat M}^{(b)}+{1\over{Z_{-}}}I_{\Rbrack0,\tau\Rbrack}\is\langle m, {\cal T}_0(M)\rangle^{\mathbb F}+\widetilde{\lambda}^{(0,\mathbb F)}I_{\Rbrack0,\tau\Rbrack}\is \langle {\cal T}_0(M) \rangle^{\mathbb F}\\
&&=S_0+{\widehat M}^{(b)}+{1\over{Z_{-}}}I_{\Rbrack0,\tau\Rbrack}\is\langle{\cal T}_b(m), {\widehat M}^{(b)}\rangle^{\mathbb G}+\widetilde{\lambda}^{(0,\mathbb F)}I_{\Rbrack0,\tau\Rbrack}\is \langle {\cal T}_b(M), {\widehat M}^{(b)} \rangle^{\mathbb G}\\
&&=S_0+{\widehat M}^{(b)}+{{{\beta^{(0,m)}}+Z_{-}\widetilde{\lambda}^{(0,\mathbb F)}}\over{{\widetilde\varphi} }}I_{\Rbrack0,\tau\Rbrack}\is \langle {\widehat M}^{(b)}\rangle^{\mathbb G}.
\end{eqnarray*}
Thus,   assertion (e) of Proposition \ref{propositionmzerotaob} implies  $I_{\{Z_{-}\geq\delta,\  \widetilde\varphi>0\}}(\lambda^{(0,\mathbb F)}Z_{-}+{\beta^{(0,m)}})/\sqrt{ \widetilde\varphi}$ belongs  to $L^2_{loc}({\cal T}_0(M),\mathbb F)$ for any $\delta>0$ iff  $I_{\Rbrack0,\tau\Rbrack}I_{\{ \widetilde\varphi>0\}}(\lambda^{(0,\mathbb F)}Z_{-}+{\beta^{(0,m)}})/  \widetilde\varphi \in L^2_{loc}(\widehat M^{(b)},\mathbb G)$ and assertion (b) follows immediately. This ends the first step.\\
{\bf Step 2.}  Herein, we prove (a) $\Longrightarrow$ (b). Suppose that  $(S^{\tau}, \mathbb G)$ satisfies SC, and denote its market's price by $\widetilde{\lambda}^{(\mathbb G)}$. Then due to $S^{\tau}=(S^{(0)})^{\tau}=S_0+ {\widehat M}^{(b)}+ (A^{(0)})^{\tau} +Z_{-}^{-1}I_{\Rbrack0,\tau\Rbrack}\is \langle m,{\cal T}_0(M)\rangle^{\mathbb F}$, and   the fact that there exists an $\mathbb F$-predictable process $\lambda$ that coincides with $ \widetilde{\lambda}^{(\mathbb G)}$ on $\Rbrack0,\tau\Rbrack$,  we obtain 
\begin{eqnarray*}
(A^{(0)})^{\tau}+Z_{-}^{-1}I_{\Rbrack0,\tau\Rbrack}\is \langle m,{\cal T}_0(M)\rangle^{\mathbb F}=\widetilde{\lambda}^{(\mathbb G)}\is \langle {\widehat M}^{(b)}\rangle^{\mathbb G}=\lambda\is \langle {\widehat M}^{(b)}\rangle^{\mathbb G}=\lambda  {\widetilde\varphi}Z_{-}^{-1}I_{\Rbrack0,\tau\Rbrack}\is \langle {\cal T}_0(M)\rangle^{\mathbb F}.
\end{eqnarray*}
Then by compensating under $\mathbb F$ on  both sides of the above equality, we obtain
\begin{eqnarray*}
Z_{-}\is A^{(0)}= - \langle m,{\cal T}_0(M)\rangle^{\mathbb F}+ \lambda  {\widetilde\varphi}\is \langle {\cal T}_0(M)\rangle^{\mathbb F}=\left(-\beta^{(0,m)}+ \lambda  {\widetilde\varphi}\right)\is \langle {\cal T}_0(M)\rangle^{\mathbb F}.
\end{eqnarray*}
and hence we conclude that
\begin{eqnarray*}
I_{\{Z_{-}\geq \delta\}}\is A^{(0)}={{-\beta^{(0,m)}+ \lambda  {\widetilde\varphi}}\over{Z_{-}}}I_{\{Z_{-}\geq \delta\}}\is \langle  {\cal T}_0(M)\rangle^{\mathbb F},
\end{eqnarray*}
and $  (-\beta^{(0,m)}+ \lambda  {\widetilde\varphi} I_{\{Z_{-}\geq \delta\}})/Z_{-} \in L^2_{loc}( {\cal T}_0(M),\mathbb F)$. This proves that $(I_{\{Z_{-}\geq \delta\}}\is S^{(0)}, \mathbb F)$ satisfies SC for any $\delta
\in (0,1)$, as well as the second equality in (\ref{G2F}) since $\lambda=\ ^{p,\mathbb F}(\widetilde{\lambda}^{\mathbb G} Z_{-}^{-1}I_{\Rbrack0,\tau\Rbrack}) $. As a consequence, the proof of (a) $\Longrightarrow $ (b) is completed. This ends the proof of the theorem.
\end{proof} 

The next theorem drops the quasi-left-continuity assumption on $S$, and gives practical sufficient conditions on $S$ and $\tau$ such that the resulting model $(S^{\tau},\mathbb G)$ fulfills the SC. 

\begin{theorem}\label{SCuptodefault} Suppose $(S,\mathbb F)$ satisfies SC  with market's price of risk denoted by  $\widetilde{\lambda}^{\mathbb F}$, and satisfies 
 \begin{equation}\label{CSbeforetaucondition}
\{\Delta M^S\neq 0\} \cap\{\widetilde{Z} =0< Z_{-} \}=\emptyset\quad\mbox{and}\quad {{(\widetilde{\lambda}^{\mathbb F}Z_{-}+{\beta^{(0,m)}})}\over{\sqrt{\widetilde\varphi}}}I_{\{\widetilde\varphi>0\}} \in L^2_{loc}(M,\mathbb F).
\end{equation}
Then $(S^{\tau},\mathbb{G})$ satisfies SC, and its market's prices of risk denoted by $\widetilde{\lambda}^{\mathbb G}$, is given by
\begin{eqnarray}\label{MarketPriceG}
\widetilde{\lambda}^{\mathbb G}:={{ Z_{-} \ ^{p,\mathbb F}(I_{\{\widetilde Z>0\}}) \widetilde{\lambda}^{\mathbb F}+{\beta^{(0,m)}}}\over{\ ^{p,\mathbb F}(I_{\{\widetilde Z>0\}}) Z_{-}\widetilde\varphi}}I_{\{ \widetilde\varphi>0\}}I_{\Rbrack0,\tau\Rbrack} .\end{eqnarray}
\end{theorem}

\begin{proof}
Let $N$ be an $\mathbb F$-local martingale.  Then we calculate
\begin{eqnarray*}
&&[{\cal T}_b(N),\widehat{M}^{(b)}]= [{\cal T}_b(N),M^{\tau}]+{\mathbb G}\mbox{-local martingale}\\
&&={{Z_{-}}\over{\widetilde Z}}I_{\Rbrack0,\tau\Rbrack}\is [N,M]+\ ^{p,\mathbb F}(\Delta N I_{\{\widetilde Z=0<Z_{-}\}})\is M^{\tau} +{\mathbb G}\mbox{-local martingale}\\
&&={{Z_{-}}\over{\widetilde Z}}I_{\Rbrack0,\tau\Rbrack}\is [N,M]+{{^{p,\mathbb F}(\Delta N I_{\{\widetilde Z=0<Z_{-}\}})}\over{Z_{-}}}I_{\Rbrack0,\tau\Rbrack}\is\langle m, M\rangle^{\mathbb F} +{\mathbb G}\mbox{-local martingale}.
\end{eqnarray*}
Then, we derive 

\begin{eqnarray}\label{NMunderG}
&&\langle{\cal T}_b(N),\widehat{M}^{(b)}\rangle^{\mathbb G}=I_{\Rbrack0,\tau\Rbrack}\is\left(I_{\{\widetilde Z>0\}}\is  [N,M]+\right)^{p,\mathbb F}   +{{^{p,\mathbb F}(\Delta N I_{\{\widetilde Z=0<Z_{-}\}})}\over{Z_{-}}}I_{\Rbrack0,\tau\Rbrack}\is\langle m, M\rangle^{\mathbb F} .
\end{eqnarray}
By applying this to the cases $N=M$ and $N=m$ respectively, we obtain 

\begin{eqnarray}\label{MMunderG}
&&\langle{\cal T}_b(M),\widehat{M}^{(b)}\rangle^{\mathbb G}=I_{\Rbrack0,\tau\Rbrack}\is\left(I_{\{\widetilde Z>0\}}\is  [M, M]\right)^{p,\mathbb F}   +{{^{p,\mathbb F}(\Delta M I_{\{\widetilde Z=0<Z_{-}\}})}\over{Z_{-}}}I_{\Rbrack0,\tau\Rbrack}\is\langle m, M\rangle^{\mathbb F} .
\end{eqnarray}

\begin{eqnarray}\label{mMunderG}
\langle{\cal T}_b(m),\widehat{M}^{(b)}\rangle^{\mathbb G}&&=I_{\Rbrack0,\tau\Rbrack}\ ^{p,\mathbb F} (I_{\{\widetilde Z>0\}})\is  \langle m,M\rangle^{\mathbb F}  +Z_{-}I_{\Rbrack0,\tau\Rbrack}\is\left( \sum \Delta M I_{\{\widetilde Z=0<Z_{-}\}}\right)^{p,\mathbb F} .
\end{eqnarray}

Thanks to $\{\Delta M\not=0\}\cap\{\widetilde Z=0<Z_{-}\}=\emptyset$, we get 
 
 \begin{eqnarray}\label{Mm/MMunderG}
\langle{\cal T}_b(m),\widehat{M}^{(b)}\rangle^{\mathbb G}=I_{\Rbrack0,\tau\Rbrack}\ ^{p,\mathbb F} (I_{\{\widetilde Z>0\}})\is  \langle m,M\rangle^{\mathbb F},\quad\quad \quad  \langle{\cal T}_b(M),\widehat{M}^{(b)}\rangle^{\mathbb G}=I_{\Rbrack0,\tau\Rbrack}\is\langle M\rangle^{\mathbb F} .\end{eqnarray}

Thus, under the assumption  that $S$ satisfies the SC$(\mathbb F)$ with its market's price denoted by $\widetilde{\lambda}^{\mathbb F}$, we get 
\begin{eqnarray*}
&&S^{\tau}=S_0+M^{\tau}+A^{\tau}=S_0+M^{\tau}+\widetilde{\lambda}^{\mathbb F}I_{\Rbrack0,\tau\Rbrack}\is\langle M\rangle^{\mathbb F}\\
&&=S_0+{\widehat M}^{(b)}+Z_{-}^{-1}I_{\Rbrack0,\tau\Rbrack}\is  \langle m,M\rangle^{\mathbb F}+\widetilde{\lambda}^{\mathbb F}I_{\Rbrack0,\tau\Rbrack}\is\langle M\rangle^{\mathbb F}\\
&&= S_0+{\widehat M}^{(b)}+Z_{-}^{-1} {^{p,\mathbb F} (I_{\{\widetilde Z>0\}})}^{-1}I_{\Rbrack0,\tau\Rbrack}\is  \langle{\cal T}_b(m),{\widehat M}^{(b)}\rangle^{\mathbb G}+ \widetilde{\lambda}^{\mathbb F}I_{\Rbrack0,\tau\Rbrack}\is\langle{\cal T}_b(M), {\widehat M}^{(b)}\rangle^{\mathbb G} .\\
&&=S_0+{\widehat M}^{(b)}+{{{\beta^{(0,m)}}Z_{-}+ \ ^{p,\mathbb F} (I_{\{\widetilde Z>0\}})\widetilde{\lambda}^{\mathbb F}}\over{\ ^{p,\mathbb F} (I_{\{\widetilde Z>0\}})Z_{-}{\widetilde\varphi}}}I_{\Rbrack0,\tau\Rbrack}\is\langle{\widehat M}^{(b)}\rangle^{\mathbb G}.
\end{eqnarray*}
The last equality is due to Theorem \ref{GKW4G}. Thus, we conclude that $(S^{\tau},\mathbb G)$ satisfies SC and its market's price  is given by (\ref{MarketPriceG})  which belongs to $L^2_{loc}(\widehat M, \mathbb G)$,  as this is equivalent to $I_{\{Z_{-}\geq\delta\}}(\lambda^{\mathbb F}Z_{-}+{\beta^{(0,m)}})/\sqrt{ \widetilde\varphi}$ belonging  to $L^2_{loc}({\cal T}_0(M),\mathbb F)$ for any $\delta>0$. This ends the proof of the theorem.
\end{proof}
\begin{remark}
 {\rm{(a)}} If either $S$ or $m$ is quasi-left-continuous (i.e. it does not jump at $\mathbb F$-predictable stopping times), then condition (\ref{CSbeforetaucondition}) is equivalent to $$\{\Delta S\neq 0\} \cap\{\widetilde{Z} =0< Z_{-} \}=\emptyset.$$
 Herein,  we put $A=A^S$ and $M=M^S$ for the sake of simplicity. Indeed, since $S$ is a special semimartingale with Doob-Meyer decomposition $S=S_0+M+A$, we get $\Delta A=\ ^{p,\mathbb F}(\Delta S)\equiv 0$, and the claim is proved. This equivalent condition appeared already in the study of the No-Unbounded-Profit-with-Bounded-Risk (NUPBR hereafter) concept under stopping with random time in \cite{ACDM}. For the financial and mathematical interpretation of this condition, we refer the reader to this paper. Thus, this connection --between the SC and the NUPBR via the above condition---  enhances our focus on the quasi-left-continuous case.\\
  {\rm{(b)}} The following process
 \begin{equation}\label{processV}
 V_t:=\sum_{0<u\leq t} I_{\{ \widetilde Z_u=0<Z_{u-}\}}\Delta M_u,\end{equation}
 is well defined, and it is a c\`adl\`ag process with finite variation. In fact it is enough to remark that there exists an $\mathbb F$-stopping time, $\widehat R$, such that $\{ \widetilde Z_u=0<Z_{u-}\}\subset \Lbrack \widehat R\Rbrack$, and $Var(V)_t\leq \vert \Delta M_{\widehat R}\vert I_{\Lbrack \widehat R,+\infty\Lbrack}.$\\
  {\rm{(c)}}  It is important to mention that, in Theorem \ref{generlaQLC}, we do not assume the SC property for $S$ nor for $S^{(0)}$, while the theorem gives a complete and precise characterization, in terms of $\mathbb F$-processes only, for the SC property to be valid for the model $\left(S^{\tau},\mathbb  G\right)$.\\
  {\rm{(b)}} The set $\{Z_{-}\geq \delta\}$ when $\delta$ varies in $(0,1)$ are intimately related to the transfer of the localization property from $\mathbb G$ to the smallest filtration $\mathbb F$. This stability property of the localization from the bigger filtration to the smallest and vice versa was established in \cite{ACDM}. 
\end{remark}
Our third main theorem of this section answers completely the problem {\bf(Prob2)}, and describes the random time models for which the SC property is preserved after stopping with $\tau$ for any model.

\begin{theorem}\label{theo:SCNULL}  For any $N\in {\cal M}^2(\mathbb F)$, we associate  the following $\mathbb F$-predictable process $\varphi_N$  given by 
\begin{eqnarray}\label{varphiN}
\varphi_N:={{d(\widetilde Z\is [N])^{p,\mathbb F}}\over{d\langle N\rangle^{\mathbb F}}}.
\end{eqnarray}
 The following assertions are equivalent.\\
  {\rm{(a)}} $\{\widetilde{Z} =0<Z_{-}\}=\emptyset$, and $\displaystyle{1\over{\varphi_N}} I_{\{Z_{-}\geq\delta,\ \varphi_N>0\}} $ is locally bounded for any $N\in {\cal M}^2(\mathbb F)$ and any $\delta>0$. \\
  {\rm{(b)}}  $\{\widetilde{Z} =0<Z_{-}\}=\emptyset$ and $\displaystyle{1\over{\varphi_N}} I_{\Rbrack 0,\tau\Rbrack}I_{\{\varphi_N>0\}}$ is $\mathbb G$-locally bounded for any $N\in {\cal M}^2_{loc}(\mathbb F)$.\\
  {\rm{(c)}} For any model $(X,\mathbb F)$ satisfying SC, the resulting model $(X^{\tau},\mathbb G)$ does also satisfy SC.
\end{theorem}
\begin{proof} This proof will be achieved in three steps. The first step proves (a)$\Longleftrightarrow$ (b). The second step deals with  (a)$\Longrightarrow$ (c), while the third step addresses  (c)$\Longrightarrow$ (a).\\
\noindent {\bf Step 1.} Let $N\in {\cal M}^2_{0,loc}(\mathbb F)$. Since $Z_{-}^{-1}I_{\Rbrack 0,\tau\Rbrack}$ is $\mathbb G$-locally bounded, we deduce the existence of a family of $\mathbb G$-stopping times $\tau_{\delta}$ that increases to infinity when $\delta$ goes to zero and 
$$
\Rbrack 0,\tau\wedge\tau_{\delta}\Rbrack\subset\{Z_{-}\geq\delta\}.$$
This implies that, on the one hand, $\displaystyle{1\over{\varphi_N}} I_{\Rbrack 0,\tau\Rbrack}I_{\{\varphi_N>0\}}$ is $\mathbb G$-locally bounded if and only if $\displaystyle{1\over{\varphi_N}} I_{\Rbrack 0,\tau_{\delta}\wedge\tau\Rbrack}I_{\{\varphi_N>0\}}$  (or equivalently $\displaystyle{1\over{\varphi_N}} I_{\Rbrack 0,\tau\Rbrack} I_{\{Z_{-}\geq\delta,\ \varphi_N>0\}}$)  is $\mathbb G$-locally bounded for any $\delta>0$. On the other hand, thanks to Proposition \ref{G/Flocalization} -(c) (see also  \cite[Proposition B.2-(c)]{ACDM}), we conclude that  $\displaystyle{1\over{\varphi_N}} I_{\Rbrack 0,\tau\Rbrack} I_{\{Z_{-}\geq\delta,\ \varphi_N>0\}}$ is $\mathbb G$-locally bounded if and only if  $\displaystyle{1\over{\varphi_N}} I_{\{Z_{-}\geq\delta,\ \varphi_N>0\}}$ is $\mathbb F$-locally bounded. This ends the first step.\\
\noindent {\bf Step 2.}  Suppose that assertion (a) holds, and let $(X,\mathbb F)$ be a model satisfying SC, where $X:=N+B$ with $N\in {\cal M}_{0,loc}^2(\mathbb F)$  and $B$ an $\mathbb F$-predictable process with finite variation. Remark that, thanks to assertion (a), the assumptions (\ref{CSbeforetaucondition}) corresponding to the model $(X,\mathbb F)$ are fulfilled. Therefore, a direct application of Theorem \ref{SCuptodefault}  implies that $(X^{\tau},\mathbb G)$ satisfies SC. This ends the second step.\\
\noindent {\bf Step 3.} Assume that assertion (c) holds, and remark that  $\{\widetilde{Z} =0<Z_{-}\} \subset \{\Delta m \neq 0\}$ which is a thin set (i.e. it is at most a union of countable graphs of $\mathbb F$-stopping times). Consider an $\mathbb F$-stopping time $T$ such that $\Lbrack  T \Rbrack \subset \{\widetilde{Z} =0<Z_{-}\} $.   Then $X^{\tau}$ satisfies SC$(\mathbb  G)$, where
\begin{eqnarray} X=I_{\Lbrack T,+\infty\Lbrack}-\left(I_{\Lbrack T,+\infty\Lbrack}\right)^{p,\mathbb F}\in {\cal M}_{0}(\mathbb F).
 \end{eqnarray}
Since $\tau<T, \ P-a.s.$ on $\{T<+\infty\}$ (due to $\widetilde Z_T=0$\ on $\{T<+\infty\}$), we deduce that
\begin{equation}\label{predictableMtau}
X^{\tau}=-I_{\Rbrack0,\tau\Rbrack} \is\left(I_{\Lbrack T,+\infty\Lbrack}\right)^{p,\mathbb F}\ \mbox{is}\ \mathbb G\mbox{-predictable and satisfies SC}(\mathbb G).\end{equation}
Hence, thanks to Lemma \ref{lem:predictableSCNUll}, we conclude that $M^\tau$ is a null process, or equivalently
\begin{eqnarray}
0=E\left(X_{\tau}\right)=E\left(\int_0^{+\infty} Z_{s-} d\left(I_{\Lbrack T,+\infty\Lbrack}\right)^{p,\mathbb F}_s\right)=E\left(Z_{T-} I_{\{ T<+\infty\}}\right).
  \end{eqnarray}
Since $Z_{T-}>0$ on  $\{T<+\infty\}$, we conclude that $T=+\infty, \ P-a.s.$, and the thin set $\{\widetilde{Z} =0<Z_{-}\}$ is evanescent (see Proposition 2.18 on Page 20 in \cite{JSlimitbook}). Thus, as a result, for any $N\in {\cal M}_{0,loc}^2(\mathbb F)$, we deduce that  
\begin{eqnarray}\label{BracketNGF}
\langle {\widehat N}^{(b)}\rangle^{\mathbb F}={{I_{\Rbrack 0,\tau\Rbrack}}\over{Z_{-}}}\is \left({\widetilde Z}\is [N,N]\right)^{p,\mathbb F}={{\varphi_N I_{\Rbrack 0,\tau\Rbrack}}\over{Z_{-}}}\is \langle N\rangle^{\mathbb F}.
\end{eqnarray}
Now we focus on proving the second claim of assertion (a). To this end, we consider $N\in {\cal M}^2(\mathbb F)$ and $\delta>0$. Consider the following process
\begin{eqnarray*}
X_{\delta}:=I_{\{Z_{-}\geq\delta\}}\is N+{{\theta-\beta^{(0,m)}}\over{Z_{-}}}I_{\{Z_{-}\geq\delta\}}\is \langle N\rangle^{\mathbb F},\quad\mbox{where}\quad  \beta^{(0,m)}:={{d\langle m, N\rangle^{\mathbb F}}\over{ d\langle N\rangle^{\mathbb F}}},\quad \theta\in L^2_{loc}(N,\mathbb F).
\end{eqnarray*}
then it is clear that $(X_{\delta},\mathbb F)$ satisfies SC, and hence $((X_{\delta})^{\tau},\mathbb G)$ satisfies SC also. As a result, by putting $N_{\delta}:=I_{\{Z_{-}\geq\delta\}}\is N$, we obtain
\begin{eqnarray*}
X_{\delta}^{\tau}&&=I_{\{Z_{-}\geq\delta\}}\is {\widehat N}^{(b)}+Z_{-}^{-1}I_{\Rbrack 0,\tau\Rbrack}\is\langle m, N_{\delta}\rangle^{\mathbb F}+ {{\theta-\beta^{(0,m)}}\over{Z_{-}}}I_{\Rbrack 0,\tau\Rbrack}I_{\{Z_{-}\geq\delta\}}\is \langle N\rangle^{\mathbb F}\\
&&=I_{\{Z_{-}\geq\delta\}}\is{\widehat N}^{(b)}+ {{\theta}\over{Z_{-}}}I_{\Rbrack 0,\tau\Rbrack}I_{\{Z_{-}\geq\delta\}}\is \langle N\rangle^{\mathbb F}\\
&&=I_{\{Z_{-}\geq\delta\}}\is{\widehat N}^{(b)}+ {{\theta}\over{\varphi_N Z_{-}}}I_{\Rbrack 0,\tau\Rbrack}I_{\{Z_{-}\geq\delta\}}\is \langle{\widehat N}^{(b)}\rangle^{\mathbb G}.
\end{eqnarray*}
The last equality follows from (\ref{BracketNGF}). Hence, we conclude that ${{\theta}\over{\varphi_N Z_{-}}}I_{\Rbrack 0,\tau\Rbrack}I_{\{Z_{-}\geq\delta\}}\in L^2_{loc}({\widehat N}^{(b)}, \mathbb G)$ for any $\theta\in L^2_{loc}(N, \mathbb F)$. Thus, by combining (\ref{BracketNGF}), Proposition \ref{G/Flocalization} -(c) and \cite[Chapter VIII 10-11]{dm2} (Lenglart's result that claims that every predictable process $H$, such that $\sup_{0\leq s\leq \cdot}\vert H_s\vert$ has a finite variation, is locally bounded), the claim ${{\theta}\over{\varphi_N Z_{-}}}I_{\Rbrack 0,\tau\Rbrack}I_{\{Z_{-}\geq\delta\}}\in L^2_{loc}({\widehat N}^{(b)}, \mathbb G)$ for any $\theta\in L^2_{loc}(N, \mathbb F)$ is equivalent to
\begin{eqnarray*}
\left({{\vert\lambda\vert}\over{\varphi_N}}I_{\{Z_{-}\geq\delta\}}\is \langle N\rangle^{\mathbb F}\right)_T<+\infty,\quad P\mbox{-a.s}.,\end{eqnarray*}
for any $\mathbb F$-predictable process $\lambda $ such that $\left(\vert\lambda\vert\is \langle N\rangle^{\mathbb F}\right)_T<+\infty,\quad P\mbox{-a.s}.$. A combination of this with \cite[Theorem 2.7]{Rudin}, we deduce that $P$-almost all $\omega\in\Omega$, the function ${1\over{\varphi_N(\omega, s)}}I_{\{Z_{s-}(\omega)\geq\delta\}},\ s\in [0,T]$ belongs to the dual of $L^1([0,T], d\langle N\rangle^{\mathbb F}_s(\omega))$ (i.e. $L^{\infty}([0,T], d\langle N\rangle^{\mathbb F}_s(\omega))$). In fact, it is enough to apply \cite[Theorem 2.7]{Rudin} to $$\Lambda_n(\lambda):=\int_0^T {{\lambda(s)}\over{\varphi_N(\omega, s)}}I_{\{Z_{s-}(\omega)\geq\delta,\ \varphi_N(\omega, s)\geq n^{-1}\}}d\langle N\rangle^{\mathbb F}_s(\omega),\quad n\geq 1,$$ which converges to $\int_0^T {{\lambda(s)}\over{\varphi_N(\omega, s)}}I_{\{Z_{s-}(\omega)\geq\delta,\ \varphi_N(\omega, s)>0\}}d\langle N\rangle^{\mathbb F}_s(\omega) $ for any $\lambda\in L^1([0,T], d\langle N\rangle^{\mathbb F}_s(\omega))$.   Therefore, $P$-almost all $\omega\in\Omega$, there exists $C(\omega)\in (0,+\infty)$ such that 
 \begin{eqnarray*}
{1\over{\varphi_N(\omega, s)}}I_{\{Z_{s-}(\omega)\geq\delta\}}\leq C(\omega),\quad \forall\ s\in ]0,T].
\end{eqnarray*}
This proves that $\displaystyle \sup_{0\leq s\leq \cdot} (\varphi_N(s))^{-1}I_{\{Z_{s-}\geq\delta,\ \varphi_N(s)>0\}}$ $\mathbb F$-predictable with finite variation. Hence, it is locally bounded due to Lenglart's result in \cite[Chapter VIII 10-11]{dm2}.This proves the theorem.
  \end{proof}


\section{Structure conditions under a class of honest times}\label{StructureConditionsAftertau}

In this section, we focus on answering the two problems {\bf(Prob1)} and {\bf(Prob2)} when we totally incorporate a random time. This can be achieved by splitting the whole half line into two stochastic intervals $\Rbrack 0, \tau\Rbrack$ and $\Rbrack \tau, +\infty\Lbrack$. The first part, i.e.  $S^{\tau}$ is already studied in the previous section. Thus, this section will concentrate on studying the Structure Conditions of $S$ on the stochastic interval  $\Rbrack \tau, +\infty\Lbrack$. Throughout this section, the random time $\tau$ will be assumed to be a honest time.  Below, we recall the definition of honest time and for more details we refer to Jeulin \cite[Chapter 4]{Jeu}.
\begin{definition}  A random time $\tau$ is called an $\mathbb F$-honest time, if for any $t$, there exists an ${\cal {F}}_t$-measurable random variable $\tau _t$ such
that $\tau I_{\{\tau<t\}}=\tau_t I_{\{\tau<t\}}$.
\end{definition}
More precisely, throughout the rest of the paper, $\tau$ is supposed to satisfy
\begin{equation}\label{conditiononTau}
\tau\ \mbox{is a honest time,}\quad  Z_{\tau}I_{\{\tau<+\infty\}}<1,\quad \mbox{and}\quad \tau<+\infty\quad P\mbox{-a.s.}.
 \end{equation}
 \begin{remark}
This assumption is also crucial for the validity of No-Unbounded-Profit-with-Bounded-Risk after an honest time. We refer to Choulli et al \cite{Choulli2013} for more details on this subject.
\end{remark}
 
 \subsection{New $\mathbb G$-local martingales for the part-after-$\tau$ and their properties}
 This subsection extends Subsections \ref{Sub3.1}-\ref{Sub3.2} to the part-after-$\tau$. Thus, we start by introducing the corresponding three transformation operators for  the part-after-$\tau$, their  properties and their interplay. This is the aim of the next two propositions.
\begin{proposition}\label{DoobDecomposition4Gafter}
Let  $X\in {\cal M}_{0,loc}({\mathbb  F })$, and  $\tau$ be  an $\mathbb F$-honest time.  Then the process ${\cal T}_1(X)$ defined by 
 \begin{eqnarray}\label{Mone}
{\cal T}_1(X):=(1-Z_{-})\is X-I_{\{\widetilde Z=1\}}\is  [X, m]+\left(I_{\{\widetilde Z=1\}}\is [X,m]\right)^{p,\mathbb F},
\end{eqnarray}
is an $\mathbb F$-local martingale,  and the processes ${\widehat{X}^{(a)}}$ and ${\cal T}_a(X)$  given by
\begin{eqnarray}
{\widehat{X}^{(a)}} &&:= I_{\Rbrack \tau, +\infty \Lbrack}\is X  +
   {I_{\Rbrack \tau, +\infty \Lbrack}}\left(1-Z_{-}\right)^{-1}\is \langle X,  m\rangle^{\mathbb F}, \label{Mafter} \\
{\cal T}_a(X)&&:=I_{\Rbrack\tau,+\infty\Lbrack}\is {\widehat X}^{(a)} +\frac{I_{\Rbrack\tau,+\infty\Lbrack}}{1-\widetilde Z} \is[X, m]- \frac{I_{\Rbrack\tau,+\infty\Lbrack}}{1-Z_{-}} \is\left( I_{\{\widetilde{Z}<1\}}\is [X,m]   \right)^{p,\mathbb F}\nonumber\\
&&=I_{\Rbrack\tau,+\infty\Lbrack}\is X +\frac{I_{\Rbrack\tau,+\infty\Lbrack}}{1-\widetilde Z} \is[X, m]+ \frac{I_{\Rbrack\tau,+\infty\Lbrack}}{1-Z_{-}} \is\left(\sum  I_{\{\widetilde{Z}=1>Z_{-}\}}(1-Z_{-})\Delta X\right)^{p,\mathbb F}\label{honestforwholeprocess}
 \end{eqnarray}
 are  $\mathbb{G}$-local martingales. Furthermore, $\widehat{X}^{(a)}\in {\cal M}_{0,lc}^2(\mathbb G)$ and ${\cal T}_1(X)\in {\cal M}^2_{0,loc}(\mathbb F)$ when $X\in {\cal M}^2_{0,loc}(\mathbb F)$.
 \end{proposition}
\begin{proof} It is clear that ${\cal T}_1(X)$ is an $\mathbb F$-local martingale  (respectively it belongs to ${\cal M}_{0,loc}^2(\mathbb F)$) as long as $X\in {\cal M}_{0,loc}(\mathbb F)$ (respectively $X\in {\cal M}_{0,loc}^2(\mathbb F)$). Let $X \in {\cal M}_{0,loc}(\mathbb F)$. The proof of ${\widehat{X}^{(a)}} \in {\cal M}_{0,loc}(\mathbb G)$ can be found in Jeulin \cite{Jeu}  (see also Barlow \cite{Barlow}). Since the $\mathbb G$-predictable process with finite variation $  {I_{\Rbrack \tau, +\infty \Lbrack}}\left(1-Z_{-}\right)^{-1}\is \langle X,  m\rangle^{\mathbb F}$ is $\mathbb G$-locally bounded,  and 
$\sup_{0\leq s\leq\cdot}(X_s-X_{s\wedge\tau})^2\leq 4\sup_{0\leq s\leq\cdot}X_{s}^2$, we deduce that ${\widehat{X}^{(a)}} \in {\cal M}_{0,loc}^2(\mathbb G)$ as soon as $X \in {\cal M}_{0,loc}^2(\mathbb F)$. Thus, by combining ${\widehat{X}^{(a)}} \in {\cal M}_{0,loc}(\mathbb G)$  and Lemma \ref{lemmecrucialapresdefault}-(b), we deduce that ${\cal T}_a(X)\in  {\cal M}_{0,loc}(\mathbb G)$. This ends the proof of the proposition.\end{proof}
The following proposition extends Propositions \ref{propositionmzerotaob}-\ref{prooTbXcompensators} to the part-after-$\tau$.
\begin{proposition}\label{Taaftertau}
Let $X, Y$ be a quasi-left-continuous elements of ${\mathcal M}_{loc}(\mathbb F)$, and let $M:=M^S$ that belongs to ${\mathcal M}^2_{loc}(\mathbb F)$ and is defined in (\ref{S-DoobMeyerDecomposition}). Then the following assertions hold.\\
{\rm{(a)}} It holds that
\begin{equation}\label{XhatYbarAfter}
[\widehat X^{(a)}, {\cal T}_a(Y)]=[\widehat Y^{(a)},  {\cal T}_a(X)]={{1-Z_{-}}\over{1-\widetilde Z}}I_{\Rbrack\tau,+\infty \Lbrack}  \is [ Y, X].
\end{equation}
Hence, $\langle\widehat X^{(a)}, {\cal T}_a(Y)\rangle^{\mathbb G}$ exists when $X,Y\in {\cal M}_{0,loc}^2(\mathbb F)$, and is given by 
\begin{eqnarray}\label{<x,Y>GafterTau}
\langle \widehat X^{(a)}, {\cal T}_a(Y)\rangle^{\mathbb G}= {{I_{\Rbrack\tau,+\infty\Lbrack}}\over{(1-Z_{-})^2}} \is\langle {\cal T}_1(Y), {\cal T}_1(X)\rangle^{\mathbb F}=I_{\Rbrack\tau,+\infty\Lbrack}\is\left(I_{\{\widetilde Z<1\}}\is [Y,X]\right)^{\mathbb F}.
\end{eqnarray}
{\rm (b)} Suppose $X\in {\cal M}_{0,loc}^2(\mathbb F)$. Then the following equalities hold
\begin{eqnarray}\label{T1Xhat}
\langle \widehat{X}^{(a)},  \widehat{X}^{(a)} \rangle^{\mathbb G} &&= \frac{I_{\Rbrack\tau,+\infty\Lbrack}}{1-Z_{-}} \is 
\left( (1-\widetilde{Z})\is [X, X] \right)^{\mathbb F},  \label{mmhatbefore}\\
{\cal T}_a\left({\cal T}_1(X)\right)&&= {\cal T}_a(X),\quad\mbox{and}\quad  \widehat{{\cal T}_1(X)}^{(a)}= {\widehat X}^{(a)}.\label{TbmzeroandTbmAfter}
\end{eqnarray}
{\rm (c)}  There exists a unique ($P\otimes d\langle {\cal T}_1(M)\rangle^{\mathbb F}$-a.e.) $\mathbb F$-predictable process $ \widetilde\psi  $ such that
\begin{eqnarray}
&&I_{\{Z_{-}<1\}}\is \left((1-\widetilde Z) \is [M,M]\right)^{p,\mathbb F}= {\widetilde\psi} \is \langle {\cal T}_1(M)\rangle^{\mathbb F},\quad 0\leq {\widetilde\psi} \leq (1-Z_{-})^{-1}I_{\{Z_{-}<1\}},\label{processphiafter1}\\
&&I_{\{\widetilde\psi=0\}}\is{\widehat M}^{(a)}\equiv 0,\quad \mbox{and}\quad  \displaystyle{{\theta}\over{\sqrt{\widetilde\psi}}}I_{\{\widetilde\psi>0\}}\in L^2_{loc}({\widehat M}^{(a)}, \mathbb G),\quad \forall
\ \theta\in L^2_{loc}(M, \mathbb F).\label{processphiafter2}
\end{eqnarray}
{\rm{(d)}}  Let $\theta^{\mathbb G}$ be a $\mathbb G$-predictable process. Then $\theta^{\mathbb G}\in L^2_{loc}({\widehat M}^{(a)}, \mathbb G)$ if and only if there exists and $\mathbb F$-predictable process $\theta$ such that $\theta^{\mathbb G}=\theta$ on $\Rbrack\tau,+\infty\Lbrack$, and 
$\theta\sqrt{\widetilde\psi}\in L^2_{loc}({\cal T}_1(M),\mathbb F)$. 

\end{proposition}
\begin{proof}
(a) Thanks to the quasi-left-continuity of $X$ and $Y$, the processes we derive   $ \frac{I_{\Rbrack \tau, \infty \Lbrack}}{1-Z_{-}}\is\langle X, \mm\rangle^{\mathbb F}$, $ \frac{I_{\Rbrack \tau, \infty \Lbrack}}{1-Z_{-}}\is\langle Y, \mm\rangle^{\mathbb F}$ and ${1_{\Rbrack \tau, \infty \Lbrack}\over{1- Z_{-}}}\is\left(I_{\{\widetilde Z<1\}}\is [Y,m]\right)^{p,\mathbb F}$ are continuous with finite variation. Hence, we derive 
\begin{eqnarray*}
[\widehat X^{(a)}, {\cal T}_a(Y)] &= & \left[X  + \frac{I_{\Rbrack \tau, \infty \Lbrack}}{1-Z_{-}}\is\langle X, \mm\rangle^{\mathbb F},
\widehat{Y}^{(a)} + {I_{\Rbrack \tau, \infty \Lbrack}\over{1 - \widetilde Z}} \is[Y, m] -
{1_{\Rbrack \tau, \infty \Lbrack}\over{1- Z_{-}}}\is\left(I_{\{\widetilde Z<1\}}\is [Y,m]\right)^{p,\mathbb F}
\right]\\
&= & I_{\Rbrack \tau, \infty \Lbrack}\is \left[X , 
 Y  +  {1\over{1 - \widetilde Z}}I_{\Rbrack \tau, \infty \Lbrack}\is[Y, m]
\right]={{1- Z_{-}}\over{1- \widetilde Z}}I_{\Rbrack \tau, \infty \Lbrack} \is [Y, X].
\end{eqnarray*}
Therefore, (\ref{XhatYbarAfter}) folows from this equality and the symmetry role of $X$ and $Y$. \\
When $X, Y\in {\cal M}_{0,loc}^2(\mathbb F)$, by combining (\ref{XhatYbarAfter}),  Lemma \ref{lemmecrucialapresdefault} , and $\langle {\cal T}_1(X)\rangle^{\mathbb F}=(1-Z_{-})^2\is \left(I_{\{\widetilde Z<1\}}\is [X, X]\right)^{p,\mathbb F}$, we easily derive (\ref{<x,Y>GafterTau}).\\
(b) Again, thanks to the quasi-left-continuity of $X$, we get  
\begin{eqnarray*}
\left[ \widehat{X}^{(a)}, \widehat{X}^{(a)}\right] = I_{\Rbrack \tau, \infty\Lbrack}\is [X,X].
\end{eqnarray*}
Thus, by applying Lemma  \ref{lemmecrucialapresdefault}-(b) to this equality, we obtain (\ref{mmhatbefore}). 
By combining (\ref{Mone}), (\ref{Mafter} ), and the quasi-left-continuity of $X$, on the one hand, we deduce that  
\begin{eqnarray*}
\widehat{{\cal T}_1(X)}^{(a)} &=& I_{\Rbrack \tau, \infty\Lbrack} \is {\cal T}_1(X)  + I_{\Rbrack \tau, \infty\Lbrack}(1- Z_{-})^{-1}\is\langle {\cal T}_1(X), m\rangle^{\mathbb F} \\
&=&   I_{\Rbrack \tau, \infty\Lbrack}\is X   +  I_{\Rbrack \tau, \infty\Lbrack}\is \left(\sum \Delta X I_{\{\widetilde Z=1>Z_{-}\}}\right)^{p,\mathbb F}
+ {{I_{\Rbrack \tau, \infty\Lbrack}}\over{1-Z_{-}}}  \is\langle X-\sum \Delta X I_{\{\widetilde Z=1>Z_{-}\}}, m\rangle^{\mathbb F}  \\
&=& \widehat{X}^{(a)} + I_{\Rbrack \tau, \infty\Lbrack}\is \left(\sum \Delta X I_{\{\widetilde Z=1>Z_{-}\}}\right)^{p,\mathbb F}
- {{I_{\Rbrack \tau, \infty\Lbrack}}\over{1-Z_{-}}} \is\langle  \sum \Delta X I_{\{\widetilde Z=1>Z_{-}\}}, m\rangle^{\mathbb F}  \\
&=& \widehat{X}^{(a)} + I_{\Rbrack \tau, \infty\Lbrack}\is \left(\sum \Delta M I_{\{\widetilde Z=1>Z_{-}\}}\right)^{p,\mathbb F}
-{{I_{\Rbrack \tau, \infty\Lbrack}}\over{1-Z_{-}}}  \is\left( \sum \Delta m \Delta X I_{\{\widetilde Z=1>Z_{-}\}} \right)^{p,\mathbb F}  \\
&=& \widehat{X}^{(a)}.
\end{eqnarray*}
On the other hand, we calculate
\begin{eqnarray*}
{\cal T}_a({\cal T}_1(X))&=& \widehat{{\cal T}_1(X)}^{(a)} + {1\over{1-\widetilde Z}}I_{\Rbrack \tau, \infty\Lbrack}\is[{\cal T}_1(X), m] - {1\over{1-Z_{-}}}I_{\Rbrack \tau, \infty\Lbrack}\is\left(I_{\{\widetilde Z<1\}}\is [{\cal T}_1(X),m]\right)^{p,\mathbb F} \\
&=& \widehat{X}^{(a)} + {1\over{1- \widetilde Z}}I_{\Rbrack \tau, \infty\Lbrack}\is \left[X-\sum \Delta X I_{\{\widetilde Z=1>Z_{-}\}}+\left(\sum \Delta X I_{\{\widetilde Z=1>Z_{-}\}}\right)^{p,\mathbb F}, m\right] \\
&& - {1\over{1-Z_{-}}}I_{\Rbrack \tau, \infty\Lbrack}\is\left(I_{\{\widetilde Z<1\}}\is \left[X-\sum \Delta X I_{\{\widetilde Z=1>Z_{-}\}}+\left(\sum \Delta X I_{\{\widetilde Z=1>Z_{-}\}}\right)^{p,\mathbb F},m \right]\right)^{p,\mathbb F} \\
&=&  \widehat{X}^{(a)}  + {1\over{1- \widetilde Z}}I_{\Rbrack \tau, \infty\Lbrack}\is \left[X,m\right] - {{I_{\Rbrack \tau, \infty\Lbrack}}\over{1-Z_{-}}} \is\left(I_{\{\widetilde Z<1\}}\is [X,m]\right)^{p,\mathbb F} = {\cal T}_a(X).
\end{eqnarray*}
These equalities prove (\ref{TbmzeroandTbmAfter}).\\
(c) Consider the following quasi-left-continuous  $\mathbb F$-local martingale
\begin{eqnarray*}
N:=(1-Z_{-})\is M - [M,m] + \langle M, m\rangle^{\mathbb F},
\end{eqnarray*}
which is locally square integrable with jumps $\Delta N = (1-\widetilde{Z} )\Delta M $. Thus, a direct application of the GKW-decomposition leads to the existence of a unique pair $(\widetilde\psi, L)\in L^2_{loc}(M_0,\mathbb F)\times {\cal M}_{0,loc}^2(\mathbb F)$ satisfying 
\begin{eqnarray}\label{GKW4N}
N={\widetilde\psi}\is {\cal T}_1(M)+L,\quad \langle {\cal T}_1(M), L\rangle^{\mathbb F}\equiv 0.\end{eqnarray}
From the definitions  of $N$ and ${\cal T}_1(M)$, we derive  
\begin{eqnarray*}
I_{\{Z_{-}<1\}}\is \left((1-\widetilde Z) \is [M,M]\right)^{p,\mathbb F}&&=I_{\{Z_{-}<1\}}\is\langle N,M\rangle^{\mathbb F}=(1-Z_{-})^{-1}I_{\{Z_{-}<1\}}\is\langle N,{\cal T}_1(M)\rangle^{\mathbb F}\\
&&={{\widetilde\psi}\over{1-Z_{-}}}I_{\{ Z_{-}<1\}}\is \langle{\cal T}_1(M),  {\cal T}_1(M)\rangle^{\mathbb F}.
\end{eqnarray*}
Since $(1-\widetilde Z)I_{\{Z_{-}<1\}} \is [M,M]$ is nondecreasing and dominated by $$I_{\{\widetilde Z<1\}}I_{\{Z_{-}<1\}}\is [M,M]=(1-Z_{-})^{-2}I_{\{ Z_{-}<1\}}\is[{\cal T}_1(M), {\cal T}_1(M)],$$ we deduce that $0\leq \widetilde\psi\leq (1-Z_{-})^{-1}I_{\{Z_{-}<1\}}$, and the proof of  (\ref{processphiafter1}) is complete. Furthermore, by using (\ref{T1Xhat}) and (\ref{processphiafter1}), on the other hand, we derive 
\begin{eqnarray*}
\langle I_{\{\widetilde\psi=0\}}\is{\widehat M}^{(a)}, I_{\{\widetilde\psi=0\}}\is{\widehat M}^{(a)}\rangle^{\mathbb G}=I_{\{\widetilde\psi=0\}}\is\langle {\widehat M}^{(a)}, {\widehat M}^{(a)}\rangle^{\mathbb G}=
 \frac{I_{\Rbrack\tau,+\infty\Lbrack}}{1-Z_{-}} I_{\{\widetilde\psi=0\}}{\widetilde\psi}\is\langle {\cal T}_1(M),  {\cal T}_1(M)\rangle^{\mathbb F}\equiv 0.
\end{eqnarray*}
On the other hand, for any $\theta\in L^2_{loc}(M, \mathbb F)$, we get 
\begin{eqnarray*}
\langle {{\theta}\over{\sqrt{\widetilde\psi}}}I_{\{\widetilde\psi>0\}}\is{\widehat M}^{(a)},  {{\theta}\over{\sqrt{\widetilde\psi}}}I_{\{\widetilde\psi>0\}}\is{\widehat M}^{(a)}\rangle^{\mathbb G}&&= 
 {{\theta^2}\over{\widetilde\psi}}I_{\{\widetilde\psi>0\}}\is\langle {\widehat M}^{(a)},  {\widehat M}^{(a)}\rangle^{\mathbb G}\\
 &&=
 {{\theta^2I_{\{\widetilde\psi>0\}}}\over{1-Z_{-}}}I_{\Rbrack \tau, \infty\Lbrack}\is\langle {\cal T}_1( M),   {\cal T}_1( M)\rangle^{\mathbb F}\leq  \theta^2\is\langle M,  M\rangle^{\mathbb F}.
 \end{eqnarray*}
 This ends the proof of assertion (c). To prove assertion (d), we  consider a $\mathbb G$-predictable process $\theta^{\mathbb G}$. Then there exists an $\mathbb F$-predictable process $\theta$ that coincides with $\theta^{\mathbb G}$ on $\Rbrack \tau, \infty\Lbrack$. Then we derive 
 \begin{eqnarray*}
 \langle \theta^{\mathbb G}\is{\widehat M}^{(a)},  \theta^{\mathbb G}\is{\widehat M}^{(a)}\rangle^{\mathbb G}&&=(\theta^{\mathbb G})^2\is\langle{\widehat M}^{(a)}, {\widehat M}^{(a)}\rangle^{\mathbb G}={{\theta^2}\over{1-Z_{-}}}I_{\Rbrack \tau, \infty\Lbrack}\is \left((1-\widetilde Z)\is [M,M]\right)^{p,\mathbb F}\\
 &&={{(\theta\sqrt{\widetilde\psi})^2}\over{1-Z_{-}}}I_{\Rbrack \tau, \infty\Lbrack}\langle {\cal T}_1( M),   {\cal T}_1( M)\rangle^{\mathbb F}.
 \end{eqnarray*}
Therefore, a combination of this equality and Lemma \ref{lemmecrucialapresdefault}-(b), assertion (d) follows immediately. 
  This ends the proof of the proposition.
\end{proof}

The next theorem investigates the GKW-decomposition of ${\cal T}_a(M)$ with respect to ${\widehat M}^{(a)}$, when $M$ belongs to $ {\cal M}_{0,loc}^2(\mathbb F)$ . It extends Theorem \ref{GKW4G}  to the ``part-after-$\tau$". 

\begin{theorem}\label{GKW4GAfter} Suppose that $M$, be given in (\ref{S-DoobMeyerDecomposition}), is quasi-left-continuous element of ${\cal M}_{0,loc}^2(\mathbb  F)$, and let ${\cal T}_1(M)$ be given by (\ref{Mzero}). Then the following assertions hold.\\
 {\rm{(a)}}  There exists $\widetilde L\in {\cal M}_{0,loc}(\mathbb G)$ such that $[\widetilde L, {\widehat M}^{(a)}]\in {\cal M}_{0,loc}(\mathbb G)$  and 
 \begin{eqnarray}\label{GKW4Tb(M)}
 \sqrt{\widetilde\psi}\is{\cal T}_a(M)={{(1-Z_{-})^{-1}}\over{\sqrt{\widetilde\psi}}} I_{\{\widetilde\psi>0\}}\is {\widehat M}^{(a)}+\sqrt{\widetilde\psi}\is\widetilde{L}.
 \end{eqnarray}
 {\rm{(b)}}  For any quasi-left-continuous $N\in {\cal M}_{0,loc}^2(\mathbb F)$, there exist unique pair $( \varphi_N, L_N)$ that belongs to  $L^2_{loc}({\cal T}_1(M),\mathbb F)\times {\cal M}_{0,loc}(\mathbb G)$ and satisfies  
 \begin{eqnarray}\label{GKW4Tb(N)}
\sqrt{\widetilde\psi}\is {\cal T}_a(N)={{(1-Z_{-})^{-1} \varphi_N}\over{\sqrt{\widetilde\psi}}}I_{\{\widetilde\psi>0\}}\is {\widehat M}^{(b)}+\sqrt{\widetilde\psi}\is L_N,\quad [L_N, {\widehat M}^{(a)}]\in {\cal M}_{0,loc}(\mathbb G),
 \end{eqnarray}
 with  \begin{eqnarray}\label{VarphiN}
 \varphi_N :=  \frac{{d\langle N, {\cal T}_1(M)}\rangle^{\mathbb F}}{d\langle {\cal T}_1(M)\rangle^{\mathbb F}}.
 \end{eqnarray}
\end{theorem}

\begin{proof}
The proof mimics exactly the proof of Theorem \ref{GKW4G} by considering ${\cal T}_{(a,n)}(N)$ given by  
\begin{eqnarray}\label{Nn}
{\cal T}_{(a,n)}(N):=\widehat{N}^{(a)} + {I_{\Rbrack  \tau, +\infty \Lbrack}\over{1 - Z_{-}}}\is{\cal T}_a(N_n), \mbox{where} \ N_n:=I_{\{1 - \widetilde Z\geq n^{-1}\}}\is[N, m]-
\left(I_{\{1-\widetilde Z\geq n^{-1}\}}\is [N,m]\right)^{p,\mathbb F},\nonumber
\end{eqnarray} 
instead. Thus, we will omit the details of this proof herein. 
\end{proof}


\subsection{Particular cases and examples for the model $(S-S^{\tau},\mathbb G)$}
This subsection treats some particular cases, such as the case when $S$ is a continuous process, and an example  of $(S, \tau)$ for which the SC feature might fails for $(S-S^{\tau},\mathbb G)$.
\begin{theorem}\label{continuouscaseAfter} Suppose that  $\tau$ fulfills (\ref{conditiononTau}). If $(S,\mathbb F)$ satisfies the SC with its market's price of risk $\widehat\lambda^{\mathbb F}$, and either $S$ or $m$ is continuous, then the the model $(S-S^{\tau},\mathbb G)$ fulfills the SC, and its market's price of risk, $\widehat\lambda^{\mathbb G}$, is given by 
\begin{eqnarray}\label{MarketPriceofRiskContinuousAfter}
{\widehat\lambda}^{\mathbb G} := \left( {\widehat\lambda}^{\mathbb F}+\frac{\beta^{(m)}}{1-Z_{-}} \right) I_{\Rbrack\tau,+\infty \Lbrack},
\quad\mbox{with}\quad \beta^{(m)}:={{d\langle m,M^S\rangle^{\mathbb F}}\over{d\langle M^S\rangle^{\mathbb F}}}.
\end{eqnarray}
 \end{theorem}
\begin{proof}
The proof is analogy to that of Theorem \ref{continuouscaseBefore} and we omit the details. 
\end{proof}

Below, we present an example where structure condition may fail. 
\begin{example}\label{example2} Herein, we present an example that when $\{\Delta X \neq 0\}\cap \{1 = \widetilde{Z} >Z_{-}\} \neq \emptyset$, $X- X^\tau$ fails to satisfy SC$(\mathbb{G})$. We suppose given a Poisson process  $N$, with intensity rate $\lambda>0$, and  ${\mathbb F}$ is the natural augmented filtration of $N$. Suppose that the stock price process $X$ is given by
$$
dX=X_{-}\sigma  dK,\, \ \ X_0=1,\ \ \ \ \ K_t:=N_t-\lambda t,$$
 or equivalently $X_t= \exp (-\lambda \sigma t +\ln (1+\sigma) N_t)$, where   $\sigma > 0$. In what follows, we introduce the notations
$$
 a:= -\frac{1}{\ln (1+\sigma)}\ln b, \ \ 0<b<1,  \ \ \mu:= \frac{\lambda \sigma}{\ln (1+\sigma)}\ \ \mbox{and}\ \  Y_t:=  \mu t -  N_t.$$
 We associate to the process $Y$ its ruin probability, denoted by
$\Psi(x)$   given by, for $x\geq 0$,  \begin{equation} \label{Psi0}\Psi(x)=P(T^x<\infty)
,\quad \mbox{ with} \quad
 T^x=\inf\{t: x+Y_t <0\}\, .\end{equation}
 
\begin{proposition} Consider the model $(X,\mathbb F)$ of Example \ref{example2}, and the following random time
\begin{equation}\label{suppoi}
\tau:=  \sup \{t:\,X_t \geq b\}=\sup \{t: Y_t \leq a \}.
\end{equation}
Then $(X- X^\tau, \mathbb{G})$ fails to satisfy SC.
\end{proposition}
\begin{proof}We recall from Aksamit et al. \cite{aksamit/choulli/deng/jeanblanc} that the supermartingale $Z$ and $m$ are given by
\begin{eqnarray*}
Z_t&=&P(\tau>t|{\cal F}_t)
 =\Psi(Y_t-a)I_{\{Y_t \geq a\}} + I_{\{Y_t < a\}}= 1+I_{\{Y_t \geq a\}}\left(\Psi(Y_t-a)-1\right),\\
\Delta  m&:=&\left(I_{\{Y_{-} > a+1\}}\phi_1  - I_{\{Y_{-}>a\}}\phi_2\right)\Delta N,\quad \phi_1:=\Psi(Y_{-}-a-1)-1,\quad \phi_2:=\Psi(Y_{-}-a)-1
\end{eqnarray*}
where $\Psi$ is defined in (\ref{Psi0}).
Then it is easy to calculate that
\begin{eqnarray*}
 \frac{1}{1-Z_{-}} I_{\Rbrack \tau, +\infty\Lbrack}\is \left \langle X, m\right \rangle_t &=&  -\lambda \sigma\int_1^t \frac{X_{u-}}{\phi_2(u)} \left\{ I_{\{Y_{u-} > a+1\}}\phi_1(u)  - I_{\{Y_{u-}>a\}}\phi_2(u) \right\}  I_{\Rbrack \tau, +\infty\Lbrack}(u) du,  \ \mbox{and}\\
 I_{\Rbrack \tau, +\infty\Lbrack}\is\left\langle \widehat{X}^{(a)}, \widehat{X}^{(a)} \right\rangle^\mathbb{G}_t &=&I_{\Rbrack \tau, +\infty\Lbrack}\frac{1}{1-Z_{-}}\is \left( (1-\widetilde{Z})\is [X,X] \right)^{p,\mathbb{F}}_t \\
 &=& \lambda \sigma^2 \int_1^t \frac{\phi_1(u) X_{u-}^2}{\phi_2(u)} I_{\{Y_{u-} >a+1\}} I_{\Rbrack \tau, +\infty\Lbrack}(u) du,
\end{eqnarray*}
where $\widehat{X}^{(a)}$ is defined via (\ref{honestforwholeprocess}). Notice that on the interval  $\{a+1 \geq Y_{-} >a\}$,
\begin{eqnarray*}
 \frac{1}{1-Z_{-}} I_{\Rbrack \tau, +\infty\Lbrack}\is \left \langle X, m\right \rangle_t &=&  \lambda \sigma\int_1^t  {X_{u-}} I_{\{Y_{u-}>a\}}   I_{\Rbrack \tau, +\infty\Lbrack} du, \ \   \mbox{while} \
 I_{\Rbrack \tau, +\infty\Lbrack}\is\left\langle \widehat{X}^{(a)}, \widehat{X}^{(a)} \right\rangle^\mathbb{G} = 0.
\end{eqnarray*}
Hence, there is no $\mathbb{G}$-predictable process $\widehat{\lambda} \in L^2_{loc}(X)$ such that
\begin{eqnarray*}
 \frac{1}{1-Z_{-}} I_{\Rbrack \tau, +\infty\Lbrack}\is \left \langle X, m\right \rangle = I_{\Rbrack \tau, +\infty\Lbrack} \widehat{\lambda}\is\left\langle \widehat{X}^{(a)}, \widehat{X}^{(a)} \right\rangle^\mathbb{G}.
\end{eqnarray*}
Hence, $X-X^\tau$ fails to satisfy SC$(\mathbb{G})$.
\end{proof}
\end{example}

\subsection{The main results for $(S-S^{\tau}, \mathbb G)$}
The following is  the first main result of this section, where we treat fully  the problem {\bf(Prob1)} for the case when $S$ is quasi-left-continuous on the stochastic interval $\Rbrack \tau, +\infty\Lbrack$.
\begin{theorem}\label{generlaQLC}
 Suppose $S$ is quasi-left-continuous. Then the following are equivalent.\\
{\rm{ (a)}} $(S- S^{\tau}, \mathbb G)$ satisfies SC with  market's prices of risk $\widetilde{\lambda}^{\mathbb G}$.\\
 {\rm{ (b)}} $(S^{(1)},\mathbb F)$ satisfies SC, and its market's price of risk $\lambda^{(1,\mathbb F)}$ satisfies $\displaystyle {{\lambda^{(1,\mathbb F)}(1-Z_{-})+{\beta^{(1,m)}}}\over{\sqrt{ {\widetilde\psi} }}}I_{\{{\widetilde\psi} >0\}}$ belongs to $ L^2_{loc}({\cal T}_1(M),\mathbb F)$, where $S^{(1)}={\cal T}_1(M)+A^{(1)}$ given by 
 \begin{equation}\label{processSOne}
 S^{(1)}:= (1-Z_{-})\is S- I_{\{ \widetilde Z=1>Z_{-}\}}\is [M^S, m],\quad 
 A^{(1)}:=(1-Z_{-})\is A-\left(I_{\{ \widetilde Z=1>Z_{-}\}}\is [M^S, m]\right)^{p,\mathbb F}.
 \end{equation}
 
 Furthermore, $\widetilde{\lambda}^{\mathbb G}$ and ${\lambda}^{(1,\mathbb F)}$ are related via the following
 \begin{eqnarray}
 \widetilde{\lambda}^{\mathbb G}= {{(1-Z_{-}) {\lambda}^{(1,\mathbb F)}+{\beta^{(1,m)}} }\over{ \widetilde{\psi}}} I_{\Rbrack \tau, +\infty\Lbrack}\ 
 \mbox{and}\ 
 {\lambda}^{(1,\mathbb F)}={{^{p,\mathbb F}(\widetilde{\lambda}^{\mathbb G}I_{\Rbrack \tau, +\infty\Lbrack}) \widetilde{\psi}- {\beta^{(1,m)}} (1-Z_{-})}\over{(1-Z_{-})^2}}I_{\{Z_{-}<1\}},
   \label{G2Fafter}   \end{eqnarray}
  where
  \begin{eqnarray}
 \quad {\beta^{(1,m)}}:={{d\langle m, {\cal T}_1(M) \rangle^{\mathbb F}}\over{d\langle {\cal T}_1(M)\rangle^{\mathbb F}}},\quad\mbox{and}\quad  \widetilde\psi:={{d\left((1-\widetilde Z) \is [M,M]\right)^{p,\mathbb F}}\over{d \langle {\cal T}_1(M)\rangle^{\mathbb F}}}.\label{F2Gafter}
  \end{eqnarray}
 \end{theorem}

\begin{proof}
The proof of this theorem is achieved in two steps. The first step proves (b)$\Longrightarrow$ (a), while the second step proves the converse. \\
{\bf Step 1.} Here we prove (b)$\Longrightarrow$ (a).   Then the quasi-left-continuity of $S$ implies the quasi-left-continuity of $M:=M^S$. Thanks to Proposition \ref{Taaftertau}, we write
\begin{eqnarray*}\label{mM/MMquasileft}
\langle{\cal T}_a({M}),\widehat{M}^{(a)} \rangle^{\mathbb G}&=&{{I_{\Rbrack\tau,+\infty\Lbrack}}\over{(1-Z_{-})^2}}\is\langle {\cal T}_1(M)\rangle^{\mathbb F},\ \mbox{and}  \
 \langle{\cal T}_a(m),\widehat{M}^{(a)}\rangle^{\mathbb G} = {{I_{\Rbrack\tau,+\infty\Lbrack}}\over{(1-Z_{-})^2}}\is  \langle m,{\cal T}_1(M)\rangle^{\mathbb F} ,\\
 {\cal T}_a({\cal T}_1(M))&=& (1-Z_{-})\is{\cal T}_a(M),\quad\mbox{and}\quad  {\widehat{{\cal T}_1(M)}}^{(a)}= (1-Z_{-})\is{\widehat M}^{(a)}.
\end{eqnarray*}
Since  $(1-Z_{-})\is (S- S^{\tau})=S^{(1)}- (S^{(1)})^{\tau}$ and $S^{(1)}$ satisfies SC,  we deduce the existence  of an $\mathbb F$-predictable process ${\lambda}^{(1,\mathbb F)}$ such that $A^{(1)}={\lambda}^{(1,\mathbb F)}\is \langle {\cal T}_1(M) \rangle^{\mathbb F}$. Hence, by combining the above equalities and Proposition \ref{Taaftertau} , we derive
\begin{eqnarray*}
(1-Z_{-})\is(S- S^{\tau})&&= ({\cal T}_1(M))^{\tau}+(A^{(1)})^{\tau}\\
&&= (1-Z_{-})\is{\widehat M}^{(a)} - {1\over{1- Z_{-}}}I_{\Rbrack \tau, +\infty\Lbrack}\is\langle m, {\cal T}_1(M)\rangle^{\mathbb F}+{\lambda}^{(1,\mathbb F)}I_{\Rbrack \tau, +\infty\Lbrack}\is \langle {\cal T}_1(M) \rangle^{\mathbb F}\\
&&= (1-Z_{-})\is{\widehat M}^{(a)} - (1-Z_{-})\is\langle{\cal T}_a(m), {\widehat M}^{(a)}\rangle^{\mathbb G}+(1-Z_{-})^2{\lambda}^{(1,\mathbb F)}\is \langle {\cal T}_a(M), {\widehat M}^{(a)} \rangle^{\mathbb G}\\
&&=  (1-Z_{-})\is{\widehat M}^{(a)}+{{{\beta^{(1,m)}}+(1-Z_{-}){\lambda}^{(1,\mathbb F)}}\over{ (1-Z_{-})^2{\widetilde\psi} }}I_{\Rbrack \tau, +\infty\Lbrack}\is \langle  (1-Z_{-})\is{\widehat M}^{(a)}\rangle^{\mathbb G}.
\end{eqnarray*}
This combined with the facts that $(1-Z_{-})^{-2}I_{\Rbrack \tau, +\infty\Lbrack}$ is $\mathbb G$-locally bounded and 
\begin{eqnarray*}\label{equivalent}
{{{\beta^{(1,m)}}+(1-Z_{-}){\lambda}^{(1,\mathbb F)}}\over{ (1-Z_{-})^2{\widetilde\psi} }}I_{\Rbrack \tau, +\infty\Lbrack}\in L^2_{loc}({\widehat M}^{(a)},\mathbb G)\ \mbox{iff}\ 
{{{\beta^{(1,m)}}+(1-Z_{-}){\lambda}^{(1,\mathbb F)}}\over{\sqrt{\widetilde\psi} }}I_{\{ \widetilde\psi>0\}}\in L^2_{loc}({\cal T}_1( M),\mathbb F),
\end{eqnarray*}
assertion (a) follows, and  the proof of  the first step is completed.\\

{\bf Step 2.}  Herein, we prove (a) $\Longrightarrow$ (b). Suppose that  $(S - S^{\tau}, \mathbb G)$ satisfies SC , and denote its market's price by $\widetilde{\lambda}^{(\mathbb G)}$. Then due to 
$$(1-Z_{-})\is (S - S^{\tau})=S^{(1)} - (S^{(1)})^{\tau}=  {\widehat {{\cal T}_1(M)}}^{(a)}+ I_{\Rbrack \tau, +\infty\Lbrack}\is A^{(1)}- (1-Z_{-})^{-1}I_{\Rbrack \tau, +\infty\Lbrack}\is \langle m,{\cal T}_1(M)\rangle^{\mathbb F},$$ and   the existence of an $\mathbb F$-predictable process $\lambda$ that coincides with $ \widetilde{\lambda}^{(\mathbb G)}$ on $\Rbrack \tau, +\infty\Lbrack$,  we obtain 
\begin{eqnarray*}
I_{\Rbrack \tau, +\infty\Lbrack}\is A^{(1)}- (1-Z_{-})^{-1}I_{\Rbrack \tau, +\infty\Lbrack}\is \langle m,{\cal T}_1(M)\rangle^{\mathbb F}&&=\widetilde{\lambda}^{(\mathbb G)}\is \langle {\widehat M}^{(a)}\rangle^{\mathbb G}=\lambda\is \langle {\widehat M}^{(a)}\rangle^{\mathbb G}\\
&&=\lambda  {\widetilde\psi }(1-Z_{-})^{-1}I_{\Rbrack \tau, +\infty\Lbrack}\is \langle {\cal T}_1(M), {\cal T}_1(M)\rangle^{\mathbb F}.
\end{eqnarray*}
Then by compensating under $\mathbb F$ on  both sides of the above equality, we obtain
\begin{eqnarray*}
(1-Z_{-})\is A^{(1)}=  \langle m,{\cal T}_1(M)\rangle^{\mathbb F}+ \lambda  {\widetilde\psi}\is \langle {\cal T}_1(M)\rangle^{\mathbb F}=\left(\beta^{(1,m)}+ \lambda  {\widetilde\psi}\right)\is \langle {\cal T}_1(M),  {\cal T}_1(M)\rangle^{\mathbb F}.
\end{eqnarray*}
and hence we conclude that
\begin{eqnarray*}
A^{(1)}=I_{\{Z_{-}<1\}}\is A^{(1)}={{\beta^{(1,m)}+ \lambda  {\widetilde\psi }}\over{1 - Z_{-}}}I_{\{Z_{-}<1\}}\is \langle   {\cal T}_1(M)\rangle^{\mathbb F},
\end{eqnarray*}
and $   I_{\{ Z_{-}<1\}} (\beta^{(1,m)}+ \lambda  {\widetilde\psi })/(1-Z_{-}) \in L^2_{loc}( {\cal T}_1(M),\mathbb F)$ which is due to the $\mathbb F$-local boundedness of $(1-Z_{-})^{-1}I_{\{ Z_{-}<1\}}$. This proves that $ (S^{(1)}, \mathbb F)$ satisfies SC, as well as the second equality in (\ref{G2Fafter}) since $\lambda=\ ^{p,\mathbb F}(\widetilde{\lambda}^{\mathbb G} (1-Z_{-})^{-1}I_{\Rbrack \tau, +\infty\Lbrack}) $. As a consequence, the proof of (a) $\Longrightarrow $ (b) is completed. This ends the proof of the theorem.
\end{proof}

\begin{remark}
It is important to mention that $(S^{(2)}, \mathbb F)$, where $S^{(2)}:= (1-Z_{-})\is S- I_{\{ \widetilde Z=1\}}\is [M^S, m]$, might not satisfy SC in general. In fact, when $(S^{(2)}, \mathbb F)$ satisfies SC, the model $(I_{\{Z_{-}=1\}}\is S^{(2)},\mathbb F)$ also satisfies SC. However, $I_{\{Z_{-}=1\}}\is S^{(2)}=-I_{\{ \widetilde Z=1=Z_{-}\}}\is [M^S, m]$ is a continuous process with finite variation. Thus, thanks to Lemma \ref{lem:predictableSCNUll}, this process should be null, which is equivalent to $m^c$ (the continuous local martingale part of $m$) being orthogonal to $I_{\{Z_{-}=1\}}\is M^S$. This latter fact might not be fulfilled in general.
\end{remark}

Now, we state the  two main theorems in this section. This answers partially the problem {\bf(Prob1)} and completely the problem {\textbf{(Prob2)}}.

\begin{theorem}\label{SCapresdefault} Suppose $\tau$ fulfills  (\ref{conditiononTau}), $S$ satisfies SC$(\mathbb F)$ with market's price of risk $\widetilde{\lambda}^{\mathbb F}$, and
\begin{equation}\label{CSaftertaucondition}
\{\Delta M^S\neq 0\}\cap \{\widetilde Z=1 >Z_{-}\}  =\emptyset,\ \mbox{and}\quad I_{\{\widetilde\psi>0\}}\left(\widetilde{\lambda}^{\mathbb F}(1-Z_{-})-\beta^{(1,m)}\right)/\sqrt{\widetilde\psi} \in L^2_{loc}({\cal T}_1(M),\mathbb F).
\end{equation}
Then $(S-S^{\tau}, \mathbb G)$ satisfies SC, and its market's prices of risk $\widetilde{\lambda}^{\mathbb G}$ is given by
\begin{eqnarray}\label{marketpriceafterTau}
\widetilde{\lambda}^{\mathbb G}={{\widetilde{\lambda}^{\mathbb F}(1-Z_{-})-\beta^{(1,m)}}\over{{\widetilde\psi}(1-Z_{-})^2}} I_{\Rbrack\tau,+\infty\Lbrack},\end{eqnarray}
where $\beta^{(1,m)}$ is given by (\ref{F2Gafter}).
\end{theorem}

\begin{proof}
Let $N$ be an $\mathbb F$-local martingale.  Then we calculate
\begin{eqnarray*}
&&[{\cal T}_a(N),\widehat{M}^{(a)} ]= [{\cal T}_a(N),M- M^{\tau}]+{\mathbb G}\mbox{-local martingale}\\
&&={{1-Z_{-}}\over{1-\widetilde Z}}I_{\Rbrack\tau,+\infty\Lbrack}\is [N,M]+ 
 \frac{I_{\Rbrack\tau,+\infty\Lbrack}}{1-Z_{-}} \is \left[  \left(  I_{\{\widetilde{Z} =1\}} \is [N,m] \right)^{p,\mathbb F}  ,M \right] +{\mathbb G}\mbox{-local martingale}\\
&&={{1-Z_{-}}\over{1-\widetilde Z}}I_{\Rbrack\tau,+\infty\Lbrack}\is [N,M] + 
 I_{\Rbrack\tau,+\infty\Lbrack}  \Delta M  \is  \left(  \sum I_{\{\widetilde{Z} = 1>Z_{-}\}} \Delta N\right)^{p,\mathbb F}+{\mathbb G}\mbox{-local martingale}.
\end{eqnarray*}
Then, we derive 
\begin{eqnarray}\label{NMunderGafter}
\langle{\cal T}_a(N),\widehat{M}^{(a)}\rangle^{\mathbb G}&&= I_{\Rbrack\tau,+\infty\Lbrack}\is  \left( I_{\{\widetilde{Z} <1\}} \is [N,M]\right)^{p,\mathbb F}\nonumber \\
&&+\ ^{p,\mathbb F} \left(\Delta MI_{\{\widetilde{Z}<1\}} \right)(1-Z_{-})^{-1}I_{\Rbrack\tau,+\infty\Lbrack} \is  \left(  \sum I_{\{\widetilde{Z} = 1>Z_{-}\}} \Delta N\right)^{p,\mathbb F}.\hskip 1cm
\end{eqnarray}
Under the first condition in (\ref{CSaftertaucondition}),  equation (\ref{NMunderGafter})  can be further simplified as 
\begin{eqnarray}\label{NMunderGSimpleafter}
&&\langle{\cal T}_a(N),\widehat{M}^{(a)}\rangle^{\mathbb G}=  I_{\Rbrack\tau,+\infty\Lbrack}\is \langle N, M \rangle^{\mathbb F}.
\end{eqnarray}

By choosing  $N=M$ and $N=m$ respectively in (\ref{NMunderGSimpleafter}), we obtain 
 \begin{eqnarray}\label{Mm/MMunderGafter}
\langle{\cal T}_a(m),\widehat{M}^{(a)}\rangle^{\mathbb G}=I_{\Rbrack\tau,+\infty\Lbrack} \is  \langle m,M\rangle^{\mathbb F},\quad  \mbox{and} \quad   \langle{\cal T}_a(M),\widehat{M}^{(a)}\rangle^{\mathbb G}=I_{\Rbrack\tau,+\infty\Lbrack}\is\langle M, M \rangle^{\mathbb F} .\end{eqnarray}

Thus, under the assumption  that $S$ satisfies the SC$(\mathbb F)$ with its market's price $\widetilde{\lambda}^{\mathbb F}$, we get 
\begin{eqnarray*}
&&I_{\Rbrack\tau,+\infty\Lbrack}\is S=I_{\Rbrack\tau,+\infty\Lbrack}\is M+I_{\Rbrack\tau,+\infty\Lbrack}\is A=I_{\Rbrack\tau,+\infty\Lbrack}\is M+\widetilde{\lambda}^{\mathbb F}I_{\Rbrack\tau,+\infty\Lbrack}\is\langle M\rangle^{\mathbb F}\\
&&={\widehat M}^{(a)}   - (1-Z_{-})^{-1}I_{\Rbrack\tau,+\infty\Lbrack}\is  \langle m,M\rangle^{\mathbb F}+\widetilde{\lambda}^{\mathbb F}I_{\Rbrack\tau,+\infty\Lbrack}\is\langle M\rangle^{\mathbb F}\\
&&= {\widehat M}^{(a)}  - (1-Z_{-})^{-1}  I_{\Rbrack\tau,+\infty\Lbrack}\is  \langle {\cal T}_a(m),{\widehat M}^{(a)}\rangle^{\mathbb G}+ \widetilde{\lambda}^{\mathbb F}I_{\Rbrack\tau,+\infty\Lbrack}\is\langle{\cal T}_a(M), {\widehat M}^{(a)}\rangle^{\mathbb G} .\\
&&={\widehat M}^{(a)} +(\widetilde{\lambda}^{\mathbb F}-  (1-Z_{-})^{-1} \beta^{(1,m)})\is I_{\Rbrack\tau,+\infty\Lbrack}\is\langle{\cal T}_a(M), {\widehat M}^{(a)}\rangle^{\mathbb G} ,\\
&&={\widehat M}^{(a)} +{{\widetilde{\lambda}^{\mathbb F}(1-Z_{-})-\beta^{(1,m)}}\over{{\widetilde\psi}(1-Z_{-})^2}} I_{\Rbrack\tau,+\infty\Lbrack}\is\langle{\widehat M}^{(a)}\rangle^{\mathbb G}.
\end{eqnarray*}
Thanks to Lemma \ref{lemmecrucialapresdefault}-(d), It is clear that
\begin{eqnarray*}
{{\widetilde{\lambda}^{\mathbb F}(1-Z_{-})-\beta^{(1,m)}}\over{{\widetilde\psi}(1-Z_{-})^2}} I_{\Rbrack\tau,+\infty\Lbrack}\in L^2_{loc}({\widehat M}^{(a)}, \mathbb G)\quad\mbox{iff}\quad I_{\{\widetilde\psi>0\}}{{\widetilde{\lambda}^{\mathbb F}(1-Z_{-})-\beta^{(1,m)}}\over{\sqrt{\widetilde\psi}}} \in L^2_{loc}({\cal T}_1(M),\mathbb F).
\end{eqnarray*}
This implies that $(S - S^{\tau},\mathbb G)$ satisfies SC and its market's price  is given by (\ref{marketpriceafterTau}), and hence the proof of the theorem is completed.
 \end{proof}

The last main result of this section gives necessary and sufficient conditions on $\tau$ that guarantee the preservation of SC for the part-after-$\tau$ for any model who posses this feature under $\mathbb F$.

\begin{theorem}\label{SCapresdefault2}
  Suppose that  $\tau$ satisfies (\ref{conditiononTau}), and for any $N\in {\cal M}_{0,loc}^2(\mathbb F)$, we put
  \begin{eqnarray}\label{psiNN}
  \psi_N:={{d\left((1-\widetilde Z)\is [N,N]\right)^{p,\mathbb F}}\over{d\langle N\rangle^{\mathbb F}}}=1-\varphi_N,
  \end{eqnarray}
 where $\varphi_N$ is defined in (\ref{varphiN}). Then the following assertions are equivalent.\\
  {\rm{(a)}} $\{\widetilde{Z} =1 > Z_{-}\}$ is evanescent and $(\psi_{N})^{-1}I_{\{\psi_N>0\}}$ is locally bounded, for any $N\in {\cal M}_{0,loc}^2(\mathbb F)$.\\
  {\rm{(b)}} $(X-X^\tau, \mathbb G)$ satisfies SC for any $(X, \mathbb F)$ satisfying SC.
\end{theorem}
\begin{proof} The proof of (a) $\Longrightarrow$ (b) is a direct consequence of Theorem \ref{SCapresdefault}. To prove the converse, we assume that assertion (b) holds, and mimic the proof of Theorem \ref{theo:SCNULL}. For the sake of completeness, we give the full details.  Since $\{\widetilde{Z} =1\}\cap \{Z_{-}<1\} \subset \{\Delta m \neq 0\}$, it is a thin set. Let  $T$ be any stopping time such that $\Lbrack  T \Rbrack \subset \{\widetilde{Z} =1\}\cap \{Z_{-}<1\} $.   Then, consider the following  $\mathbb F$-martingale,
\begin{eqnarray} M=V-\widetilde V\in {\cal M}_{0}(\mathbb F),\ \mbox{where} \ \ V:=I_{\Lbrack T,+\infty\Lbrack}\ \ \mbox{and} \  \ \widetilde V:=(V)^{p,\mathbb F}.
 \end{eqnarray}
Since $\{t>\tau\} \subset\{\widetilde{Z}_t <1\}$ (see Jeulin \cite{Jeu} or Choulli et al. \cite{Choulli2013}) and $\widetilde{Z}_T =1$ on $\{T<+\infty\}$, we deduce that $\tau \geq T, \ P-a.s.$, and
\begin{equation}\label{predictableMtauafter}
M - M^{\tau}= - I_{\Rbrack \tau, +\infty \Lbrack}\is \widetilde V
\end{equation}
satisfies SC($\mathbb G$). By combining this with Lemma \ref{lem:predictableSCNUll}, we conclude that $M - M^\tau$ must be null (i.e. $ I_{\Rbrack \tau, +\infty \Lbrack}\is \widetilde V=0$). Thus, we get
\begin{eqnarray}
0=E\left( I_{\Rbrack \tau, +\infty \Lbrack}\is \widetilde V_{\infty}\right)=E\left(\int_1^{+\infty}(1-  Z_{s-} )d\widetilde V_s\right)=E\left( (1 - Z_{T-}) I_{\{ T<+\infty\}}\right),
  \end{eqnarray}
  or equivalently $ (1 - Z_{T-}) I_{\{ T<+\infty\}}=0$ which implies that  $T=+\infty, \ P-a.s.$.   Therefore, the thin set $\{\widetilde{Z} =1\}\cap \{Z_{-}<1\}$ is evanescent (see Proposition 2.18 on Page 20 in \cite{JSlimitbook}).  Thus, as a result, for any $N\in {\cal M}_{0,loc}^2(\mathbb F)$, we deduce that  
\begin{eqnarray}\label{BracketNGFafter}
\langle {\widehat N}^{(a)}\rangle^{\mathbb F}={{I_{\Rbrack \tau, +\infty \Lbrack}}\over{1-Z_{-}}}\is \left((1-{\widetilde Z})\is [N,N]\right)^{p,\mathbb F}={{\psi_N I_{\Rbrack \tau, +\infty \Lbrack}}\over{1-Z_{-}}}\is \langle N\rangle^{\mathbb F}.
\end{eqnarray}
Therefore, the rest of this proof focus on  the second statement of assertion (a), we consider $N\in {\cal M}^2(\mathbb F)$, and defined the following processes
\begin{eqnarray*}
X:=I_{\{Z_{-}<1\}}\is N+{{\theta+\beta^{(N,m)}}\over{1-Z_{-}}}I_{\{Z_{-}<1\}}\is \langle N\rangle^{\mathbb F},\quad\mbox{where}\quad  \beta^{(N,m)}:={{d\langle m, N\rangle^{\mathbb F}}\over{ d\langle N\rangle^{\mathbb F}}},\quad \theta\in L^2_{loc}(N,\mathbb F).
\end{eqnarray*}
Then it is clear that $(X,\mathbb F)$ satisfies SC, and hence $(X^{\tau},\mathbb G)$ satisfies SC also. As a result, since $\Rbrack \tau, +\infty \Lbrack\subset\{Z_{-}<1\}$, we obtain
\begin{eqnarray*}
X^{\tau}&&= {\widehat N}^{(a)}-(1-Z_{-})^{-1}I_{\Rbrack \tau, +\infty \Lbrack}\is\langle m, N\rangle^{\mathbb F}+ {{\theta+\beta^{(N,m)}}\over{1-Z_{-}}}I_{\Rbrack\tau,+\infty\Lbrack}\is \langle N\rangle^{\mathbb F}\\
&&={\widehat N}^{(a)}+ {{\theta}\over{1-Z_{-}}}I_{\Rbrack \tau, +\infty \Lbrack}\is \langle N\rangle^{\mathbb F}\\
&&={\widehat N}^{(a)}+ {{\theta}\over{\psi_N }}I_{\{ \psi_N >0\}}\is \langle{\widehat N}^{(b)}\rangle^{\mathbb G}.
\end{eqnarray*}
The last equality follows from (\ref{BracketNGF}). Hence, we conclude that ${{\theta}\over{\psi_N }}I_{\{ \psi_N >0\}}\in L^2_{loc}({\widehat N}^{(b)}, \mathbb G)$ for any $\theta\in L^2_{loc}(N, \mathbb F)$. Thus, by combining (\ref{BracketNGF}), Proposition \ref{G/Flocalization} -(c) and \cite[Chapter VIII 10-11]{dm2} (Lenglart's result that claims that every predictable process $H$, such that $\sup_{0\leq s\leq \cdot}\vert H_s\vert$ has a finite variation, is locally bounded), the claim ${{\theta}\over{\psi_N }}I_{\{ \psi_N >0\}}\in L^2_{loc}({\widehat N}^{(b)}, \mathbb G)$ for any $\theta\in L^2_{loc}(N, \mathbb F)$ is equivalent to
\begin{eqnarray*}
\left({{\vert\lambda\vert}\over{\psi_N}}I_{\{\psi_N>0,\ Z_{-}<1\}}\is \langle N\rangle^{\mathbb F}\right)_T<+\infty,\quad P\mbox{-a.s}.,\end{eqnarray*}
for any $\mathbb F$-predictable process $\lambda $ such that $\left(\vert\lambda\vert\is \langle N\rangle^{\mathbb F}\right)_T<+\infty,\quad P\mbox{-a.s}.$. A combination of this with \cite[Theorem 2.7]{Rudin}, we deduce that $P$-almost all $\omega\in\Omega$, the function ${1\over{\varphi_N(\omega, s)}}I_{\{\psi_N>0,\ Z_{s-}(\omega)<1\}},\ s\in [0,T]$ belongs to the dual of $L^1([0,T], d\langle N\rangle^{\mathbb F}_s(\omega))$ (i.e. $L^{\infty}([0,T], d\langle N\rangle^{\mathbb F}_s(\omega))$). In fact, it is enough to apply \cite[Theorem 2.7]{Rudin} to $$\Lambda_n(\lambda):=\int_0^T {{\lambda(s)}\over{\psi_N(\omega, s)}}I_{\{Z_{s-}(\omega)<1,\ \psi_N(\omega, s)\geq n^{-1}\}}d\langle N\rangle^{\mathbb F}_s(\omega),\quad n\geq 1,$$ which converges to $\int_0^T {{\lambda(s)}\over{\psi_N(\omega, s)}}I_{\{Z_{s-}(\omega)<1,\ \psi_N(\omega, s)>0\}}d\langle N\rangle^{\mathbb F}_s(\omega) $ for any $\lambda\in L^1([0,T], d\langle N\rangle^{\mathbb F}_s(\omega))$.   Therefore, $P$-almost all $\omega\in\Omega$, there exists $C(\omega)\in (0,+\infty)$ such that 
 \begin{eqnarray*}
{1\over{\psi_N(\omega, s)}}I_{\{\psi_N(\omega,s)>0,\ Z_{s-}(\omega)<1\}}\leq C(\omega),\quad \forall\ s\in ]0,T].
\end{eqnarray*}
This proves that $\displaystyle \sup_{0\leq s\leq \cdot} (\psi_N(s))^{-1}I_{\{Z_{s-}<1,\ \psi_N(s)>0\}}$ $\mathbb F$-predictable with finite variation. Hence, it is locally bounded due to Lenglart's result in \cite[Chapter VIII 10-11]{dm2}.This proves assertion (a), and completes the proof of the theorem.
\end{proof}

\begin{remark}
{\rm{(a)}} When $m$ is continuous, we get $\widetilde Z=Z_{-}$ and $\psi_N=1-Z_{-}$ for any $N\in {\cal M}_{0,loc}^2(\mathbb F)$. Hence in this case, $(X-X^{\tau}, \mathbb G)$ always satisfies SC whenever $(X,\mathbb F)$ does.\\
{\rm{(b)}} It is clearly that the spirit of Example \ref{CS96}  appears naturally in Theorem {SCapresdefault2}. \end{remark}




\bigskip \bigskip

{\bf Acknowledgements:} An important part of  this research was achieved at University of Alberta. The research of T. Choulli  is supported financially by the
Natural Sciences and Engineering Research Council of Canada (Grant G121210818). The research of Jun Deng is supported by the National Natural Science Foundation of China (11501105) and UIBE Excellent Young Research Funding (302/871703).

\appendix 

\section{Useful results}

\begin{proposition}\label{G/Flocalization} The following assertions hold.\\
  {\rm{(a)}} There exists a sequence of $\mathbb G$-predictable stopping times, $(\tau_n)_{n\geq 1}$, that increases to $+\infty$ and 
  \begin{eqnarray}
  \Rbrack0,\tau_n\wedge\tau\Rbrack\subset \{Z_{-}\geq n^{-1}\}.
  \end{eqnarray}
  {\rm{(b)}} Let $(\sigma^{\mathbb G}_n)_n$ be a sequence of $\mathbb G$-predictable stopping times that increases to $+\infty$.
 Then, there exists a non-decreasing sequence of $\mathbb F$-predictable stopping times, $(\sigma^{\mathbb F}_n)_{n\geq 1}$,
 satisfying the following properties
  \begin{eqnarray}
  \sigma^{\mathbb G}_n\wedge\tau= \sigma^{\mathbb F}_n\wedge\tau,\ \ \ \sigma_{\infty}^{\mathbb F}:=\sup_{n} \sigma^{\mathbb F}_n\geq  R\ \ P-a.s.,\label{G/Fstoppingbeforetau1} \\
  \ \ \mbox{and }\ \ \ \ \ \ \ \ Z_{\sigma_{\infty}^{\mathbb F}-}=0\ \ \ P-a.s.\ \ \ \ \mbox{on}\ \ \ \ \Sigma\cap(\sigma_{\infty}^{\mathbb F}<+\infty),\label{G/Fstoppingbeforetau2}
  \end{eqnarray}
  where $\Sigma:=\displaystyle\bigcap_{n\geq 1}(\sigma_n^{\mathbb F}<\sigma_{\infty}^{\mathbb F})$ and $R:=\inf\{t\geq 0\ |\ Z_t=0\}$.\\
   {\rm{(c)}} Let $V$ be an $\mathbb F$-predictable and non-decreasing process with values in $[0,+\infty]$.
   Then, $V^{\tau}\in{\cal A}_{loc}^+(\mathbb G)$ if and only if $I_{\{ Z_{-}\geq\delta\}}
   \is V\in {\cal A}_{loc}^+(\mathbb F)$ for any $\delta>0$.\\
\end{proposition}


\begin{lemma}[\cite{Choulli2013}]\label{lemmecrucial}
The following assertions hold.\\
{\rm{(b)}}  For any $\mathbb F$-adapted process $V$  with locally integrable
variation, we have
\begin{equation}\label{GcompensatorofVbeforetaugeneral}
\comg{V^{\tau}} =(Z_{-})^{-1}I_{\Lbrack 0,\tau\Rbrack}\is\bigl(\widetilde Z \is V\bigr)^{p,\mathbb F}  . \end{equation}
{\rm{(b)}}  For any $\mathbb F$-local martingale  $M$, we have, on $\Lbrack 0,\tau\Rbrack$
\begin{equation}\label{ZZtildeinaccessiblejumps}
\prog{\frac{\Delta M}{\widetilde Z}}={{ \prof{{\Delta M}I_{\{\widetilde Z>0\}}}}\over{Z_{-}}},\ \ \mbox{and}\ \ \prog{\frac{1}{\widetilde Z}}={{\prof{I_{\{\widetilde Z>0\}}}}\over{Z_{-}}}.
\end{equation}
\end{lemma}

\begin{lemma}\label{lem:vb}
The following process
 \begin{equation}\label{eq:pfzlocalbounded}
  V^{(b)} := \left(^{p,\mathbb{F}}\left( I_{\{\widetilde{Z} >0\}} \right)\right)^{-1} I_{\Lbrack 0, \tau \Rbrack}
 \end{equation}
 is $\mathbb{G}$-predictable and locally bounded.
\end{lemma}
\begin{proof}
 It is enough to notice that $\widetilde{Z} \leq  I_{\{\widetilde{Z} >0\}} $ and the process $(Z_{-})^{-1} I_{\Lbrack 0, \tau \Rbrack} $ is $\mathbb{G}$-locally bounded.
\end{proof}

\begin{lemma}\label{lemmecrucialapresdefault} Suppose that $\tau$ is a honest time.  Then the following assertions hold.\\
{\rm{(a)}}If $Z_\tau I_{\{\tau<+\infty\}} <1$ $P$-a.s., then $({1-Z_{-}})^{-1}I_{\Rbrack\tau,+\infty\Lbrack}$ is a $\mathbb G$-locally bounded and predictable process.\\
{\rm{(b)}} For any $\mathbb F$-adapted process with locally integrable variation, $V$, we have
\begin{eqnarray}\label{Gcompensatoraftertau}
  I_{\Rbrack \tau,+\infty\Lbrack}\is{ V}^{p,\mathbb G} &=&  I_{\Rbrack \tau,+\infty\Lbrack} \left({1-Z_{-}}\right)^{-1} \is \left((1- \widetilde{Z}) \is V\right)^{p,\mathbb F}.
 \end{eqnarray}
{\rm{(c)}}  Suppose $\tau$ is finite almost surely and $Z_\tau<1$ $P$-a.s.. Then, $I_{\Lbrack \tau,+\infty\Rbrack}\is V\in {\cal A}_{loc}(\mathbb G)$
 if and only if $(1-\widetilde Z)\is V\in {\cal A}_{loc}(\mathbb F)$.\\
 {\rm{(d)}}  Suppose $\tau$ is finite almost surely  and $Z_\tau<1$ $P$-a.s., and $V$ is a nondecreasing and  $\mathbb
F$-predictable process. Then, for any $\mathbb
F$-predictable process $\varphi\geq0$,\\ $\varphi I_{\Rbrack
\tau,+\infty\Lbrack}\is V\in {\cal A}^+_{loc}(\mathbb G)$ iff $(1-Z_{-})\varphi\is V\in {\cal
A}^+_{loc}(\mathbb F)$ iff $\varphi I_{\{Z_{-}<1\}}\is V\in {\cal
A}^+_{loc}(\mathbb F)$.
\end{lemma}

\begin{lemma}
  Suppose that $\tau$ is a honest time and $Z_\tau I_{\{\tau<+\infty\}} <1$ $P$-a.s.. Then the process
 \begin{equation}\label{eq:pf1-zlocalbounded}
  V^{(b)} :=  \left({^{p,\mathbb{F}}\left( I_{\{\widetilde{Z} <1\}} \right)}\right)^{-1} I_{\Rbrack \tau,+\infty \Lbrack}
 \end{equation}
 is $\mathbb{G}$-predictable and locally bounded.
\end{lemma}
\begin{proof}
 It is enough to notice that $1 -\widetilde{Z} \leq  I_{\{\widetilde{Z} <1\}} $ and the process $(1-Z_{-})^{-1}I_{\Rbrack \tau,+\infty\Lbrack} $ is $\mathbb{G}$-locally bounded.
\end{proof}





\end{document}